\documentclass[aap,preprint]{imsart}

\RequirePackage[OT1]{fontenc}
\RequirePackage{amsthm,amsmath}
\startlocaldefs
\numberwithin{equation}{section}
\theoremstyle{plain}
\endlocaldefs

\usepackage[utf8]{inputenc}
\usepackage{amsfonts}
\usepackage{amssymb}
\usepackage{array}
\usepackage{bm}
\usepackage{multirow}
\usepackage{subfig}
\usepackage{enumitem}
\usepackage[normalem]{ulem}
\usepackage[font=small]{caption}
\usepackage{graphicx}
\usepackage{mathtools}
\usepackage{xspace}

\usepackage{tikz}
\usetikzlibrary{arrows,matrix,positioning,fit,calc}

\RequirePackage[colorlinks,citecolor=blue,urlcolor=blue]{hyperref}

\usepackage{algorithm}
\usepackage{algpseudocode}

\algnewcommand{\IIf}[1]{\State\algorithmicif\ #1\ \algorithmicthen}
\newtheorem{theorem}{Theorem}

\newtheorem{lemma}{Lemma}

\newtheorem{remark}{Remark}

\makeatletter
\def\blx@maxline{77}
\makeatother

\newcommand{\ex}{\mathbb E}

\newcommand{\pr}{\mathbb P}

\newcommand{\var}{\text{var}}
\newcommand{\ones}{\mathbf 1}

\newcommand{\tI}{\tilde I}
\newcommand{\tS}{\tilde S}

\begin{document}
	
	\begin{frontmatter}
		\title{Adaptive Randomization in Network Data}
		\runtitle{Adaptive Randomization in Network Data}
		\begin{aug}
	\author[A]{\fnms{Zhixin}\snm{Zhou}, \ead[label=e1]{zhixzhou@cityu.edu.hk}}
	\author[B]{\fnms{Ping} \snm{Li}\ead[label=e2]{liping98@baidu.com}}
	\and
	\author[C]{\fnms{Feifang}\snm{Hu}
		\ead[label=e3]{feifang@email.gwu.edu}}
	\address[A]{
		Department of Management Sciences, City University of Hong Kong,
		\printead{e1}}

	\address[B]{Baidu Research, USA,
		\printead{e2}}

	\address[C]{
		Department of Statistics, George Washington University,
		\printead{e3}}
\end{aug}
		
		\begin{abstract}
			Network\footnotetext{The work of Zhixin Zhou was conducted while he was a Postdoctoral Researcher at Baidu Research - Bellevue. The work of Feifang Hu was conducted while he was a consulting researcher at Baidu Research.} data have appeared frequently in recent research. For example, in comparing the effects of different types of treatment, network models have been proposed to improve the quality of estimation and hypothesis testing. In this paper, we focus on efficiently estimating the average treatment effect using an adaptive randomization procedure in networks. We work on models of causal frameworks, for which the treatment outcome of a subject is affected by its own covariate as well as those of its neighbors. Moreover, we consider the case in which, when we assign treatments to the current subject, only the subnetwork of existing subjects is revealed. New randomized procedures are proposed to minimize the mean squared error of the estimated differences between treatment effects. In network data, it is usually difficult to obtain theoretical properties because the numbers of nodes and connections increase simultaneously. Under mild assumptions, our proposed procedure is closely related to a time-varying inhomogeneous Markov chain. We then use Lyapunov functions to derive the theoretical properties of the proposed procedures. The advantages of the proposed procedures are also demonstrated by extensive simulations and experiments on real network data.
		\end{abstract}
		
		
	\end{frontmatter}
	
	\section{Introduction}\label{sec:intro}
	
	Evaluation of the effects of different types of treatment is gaining significant attention in social media development, online advertising and clinical testing. The outcome for each subject may depend not only on the treatment allocation, but also the subjects' covariates and the connections between subjects. Random treatment assignment methods often generate unbalanced prognostic factors. In the situation where the covariates are the observed categorical or numerical variables in fixed dimensions, sequential treatment assignment is introduced in~\cite{pocock1975sequential} to address the issue of unbalancedness. In~\cite{wei1978application}, the author generalizes the idea of sequential design by proposing a marginal urn model.  Adaptive randomization methods are studied in~\cite{hu2012asymptotic, hu2012balancing}, and show promising performance in categorical covariate balance with theoretical guarantees. Pairwise sequential randomization is investigated in~\cite{qin2016pairwise} to reduce the Mahalanobis distance of continuous variables.
	
	In the past decade, the presence of networks in social media, clinical tests and biological experiments has received attention from statisticians~\cite{wasserman1994social, borgatti2009network, carrington2005models, bickel2009nonparametric}. In causal inference studies, the behavior of one individual may be correlated with the behaviors of other individuals, namely peer effects or social interaction~\cite{manski2000economic, aral2011creating, eckles2017design}. In online social media networks, the behavior of a given user may be similar to his or her friends, as they might share correlated factors. Hence, in causal inference and clinical studies, we assume that if two subjects are connected in the network, then their hidden covariates affect each other's outcomes. To be more precise, we consider \emph{network-correlated outcomes}, where the network informs the correlations among potential outcomes because the potential outcomes of subject $i$ depend on both its own covariates and those of its neighbors in the network~\cite{manski2013identification, basse2018model}. Furthermore, we assume the potential outcome of a certain subject is not affected by the assignment of treatments to other subjects~\cite{cox1958planning}. That is, there is no interference between subjects~\cite{aral2016networked}. In addition, we consider another realistic assumption, which is similar to that proposed in~\cite{wei1978application}: we assume subjects appear singly and must be treated immediately. In other words, when we decide the treatment assigned to the current subject, only the connections between this subject and the previous subjects are observed; we observe only the sub-adjacency matrix for those subjects observed in the current stage. Rerandomization is proposed in~\cite{morgan2012rerandomization} and generalized to network data by~\cite{basse2018model}; however, their approach requires the whole network to be revealed before deciding the treatment of the first subject. To resolve this issue, here, we generalize adaptive design methods~\cite{hu2012asymptotic, hu2012balancing} to decide treatment allocation sequentially. It is worth noting that the adaptive design method has not previously been considered in network models. Moreover, the performance analysis of the existing adaptive randomization method cannot be applied to the model considered in this paper.
	
	Assuming the observations are network-correlated and sequentially obtained, this paper focuses on improving the estimation of treatment effects by reducing the \emph{imbalance measurement}. We still aim to reduce the effect of prognostic factors by the pairwise sequential randomization method proposed in~\cite{qin2016pairwise}. Under the assumption of network-correlated outcomes, and supposing the network is observed sequentially, we first derive the formula for variance of treatment effects under certain statistical assumptions, then we show that our approach reduces the imbalance measurement empirically and theoretically under some reasonable assumptions on the network. Despite the popularity of the model in~\cite{manski2013identification, basse2018model}, no previous work has analytically evaluated the variance of the estimator in this model with mathematical verification. To the best of our knowledge, this paper is the first work to provide a theoretical verification for the performance of randomization procedures on models assuming network-correlated outcomes.
	
	In the literature, it is assumed that covariates are identically and independently distributed (i.i.d.), and the number of covariates is fixed, hence turning the imbalance measurement of the adaptive randomization procedure into a Markov process. It is shown in~\cite{hu2012asymptotic} that the Markov process is recurrent when the covariates are categorical variables.  To formulate a theoretical analysis of the proposed procedure of this paper, we assume the observed network follows the Erd\H os-R\' enyi random graph model. The analysis does not follow from previous work on adaptive design, in the following sense. As we observe a network with extra nodes, the number of possible neighbors of each individual increases simultaneously. Moreover, because the Erd\H os-R\' enyi random graph is a probabilistic model for undirected graphs, the entries of the adjacency matrix are not independent. To overcome these difficulties, we analyze this stochastic process as a Lamperti problem~\cite{lamperti1960criteria} and further derive the upper bound of the expectation of imbalance measurement by computing certain Lyapunov functions~\cite{menshikov2016non}. In our model, as more and more subjects join the experiments, the dimension of states changes over time progresses. Thus, this process can be approximated as a time-varying Markov process. The generalization from fixed dimension to increasing dimension is a novel extension in Markov models.\par
	
	This article is organized as follows. We introduce the network-correlated outcome model and our proposed procedure in Section~\ref{sec:model}. Theoretical properties under the Erd\H os-R\' enyi random graph model are presented in Section~\ref{sec:theory}. In Section~\ref{sec:gaussian}, we discuss the theoretical properties that arise when we replace the random graph model with a Gaussian orthogonal ensemble. Experiments on simulated and real network data are presented in Section~\ref{sec:exp}. We conclude in Section~\ref{sec:conclusion}, where possible future works are also discussed. Proofs of the main theorems and auxiliary lemmas appear in Section~\ref{sec:proof}.
	
	\newpage
	
	{Here, we briefly introduce the notation used in this paper. $ X^n$ is the set of vectors with entries belonging to $X$, where $X$ can be any subset of real numbers. Similarly, $X^{m\times n}$ is the set of $m\times n$ matrices with entries belonging to $X$. For $A\in S^{m\times n}$, $A_{i*}\in X^{n}$ is the $i$-th row of matrix $A$. For vector $a$, $\|a\|$ represents the $\ell^2$-norm of vector $a$. $a_{i:j} =(a_i, a_{i+1}, \dots, a_{j})$ for $i<j$. Similarly, for matrix $A$, $A_{i:j, k:l}$ is the submatrix formed by rows $i, i+1, \dots, j$ and columns $k, k+1, \dots, l$. In particular, we write $A^{(i)}=A_{1:i, 1:i}$ as the upper-left submatrix. }

	\section{Model Assumptions}\label{sec:model}
	
	We focus on two treatment groups (treatment 0 and treatment 1) assigned to a finite population of $n$ subjects.  Let $T\in \mathbb \{0,1\}^n$ be the treatment assignment vector. $T_i$ records the assignment of the $i$-th subject, that is, $T_i=0$ for treatment $0$ and $T_i=1$ for treatment $1$. The relationship between nodes is recorded by an undirected network, or equivalently, a symmetric binary adjacency matrix $A\in\{0, 1\}^{n\times n}$. We assume self loops always exist, i.e., $A_{ii} = 1$ for $i \in [n]$. We recall that $A_{i*}$ is the $i$-th row of adjacency matrix $A$. Given the treatment assignment $T_i$, the observed outcome of the $i$-th subject follows the distribution
	\begin{align}\label{eq:model:assump}
	X_i = \mu_0 (1-T_i) + \mu_1 T_i + A_{i*}Z + \varepsilon_i
	\hspace{0.13in} \text{where} \hspace{0.13in}
	Z\sim \mathcal N(0, \sigma_Z^2 I_{n}) \text{ and } \varepsilon_i\sim \mathcal N(0, \sigma_\varepsilon^2).
	\end{align}
	We assume $\varepsilon_i$ are i.i.d. for $i\in[n]$. The observation is the summation of three parts.
	\begin{itemize}
		\item[1.] $\mu_0 (1-T_i) + \mu_1 T_i $ is the treatment effect, where $\mu_0$ and $\mu_1$ are the effect sizes of the corresponding treatments. We note that the outcome has the expectation $\ex[X_i] = \mu_0$ if $T_i=0$, otherwise its expectation is $\mu_1$.
		\item[2.] The outcome of the $i$-th observation is also affected by its unknown covariate $Z_i$ and the covariates of its neighbors in the network. To be precise, let $N_i$ be the set of neighbors of $i$, and recall that $A_{ii}=1$, then $A_{i*}Z = Z_i+\sum_{j: j\in N_i} Z_j$. We assume the covariates $Z$ have zero mean, {so the outcome can be positively or negatively influenced by the covariates}.
		\item[3.] $\varepsilon_i$ is random noise in each observation. We also write $\varepsilon:=(\varepsilon_1, \dots, \varepsilon_n)^\top$, which follows the distribution $\mathcal N(0, \sigma_{\varepsilon}^2 I_n)$.
	\end{itemize}

	Following previous studies, to ensure the treatment groups are unbiased, we restrict $T_{2m-1}+T_{2m}=1$. That is,  $(T_{2m-1}, T_{2m})$ is either $(0,1)$ or $(1,0)$. For notational convenience, we assume the total number of subjects $n$ is even. Hence we have an estimator of $\mu_0-\mu_1$, defined as
	\begin{align*}
	W:=&\frac 2n \sum_{i=1}^n (1-T_i)X_i-T_iX_i
	= \mu_0 - \mu_1 + \frac 2n \sum_{i=1}^n (1-T_i) (A_{i*}Z+\varepsilon_i) - T_i (A_{i*} Z+\varepsilon_i)  \\
	=& \mu_0 - \mu_1 + \frac 2n ({\bf 1}_n-2T)^\top (AZ+\varepsilon).
	\end{align*}
	For a fixed adjacency matrix $A$ and an allocation vector $T$, it is not difficult to check that the estimator is unbiased, as $Z$ and $\varepsilon$ have zero means:
	\begin{align*}
	\ex[W] = \mu_0 - \mu_1 +  \frac 2n ({\bf 1}_n-2T)^\top \ex[AZ+\varepsilon]
	= \mu_0 - \mu_1 +  \frac 2n ({\bf 1}_n-2T)^\top A (\ex[Z]+\ex[\varepsilon])
	=\mu_0-\mu_1.
	\end{align*}
	We can also compute the variance of $W$:
	\begin{align*}
	\text{var}[W]
	= \frac 4{n^2}\var[ ({\bf 1}_n-2T)^\top (AZ+\varepsilon) ]
	=\frac 4{n^2}\| A({\bf 1}_n-2T)\|^2 \sigma^2_Z + \frac 4n \sigma_{\varepsilon}^2,
	\end{align*}
	where $\| \cdot \|$ denotes the $\ell^2$-norm throughout this paper. We note that $W$ is an unbiased estimation and the term $4\sigma_{\varepsilon}^2/n$ converges to 0 as $n\to\infty$, so the best strategy in this experiment is to reduce the term $\| A({\bf 1}_n-2T)\|^2$ by assigning an appropriate treatment to each pair of subjects. As the variance of estimator $W$ decreases, the hypothesis testing on the effectiveness of the treatment becomes more powerful. We assume each pair of subjects joins the experiment sequentially, and we need to decide their treatment assignment soon after they join. In pairwise sequential randomization~\cite{qin2016pairwise, ma2019statistical}, we assign different treatments to each pair of subjects simultaneously. In the $m$-th stage, we determine the treatment assignments to the $(2m-1)$-th and the $2m$-th subjects, which may depend on two factors. First, after the first $2m$ subjects join the experiment, we only observe the connection between these subjects, while all other connections are concealed. In other words, we observe the (upper-left) sub-adjacency matrix $A^{(2m)}:=(A_{ij})_{1\le i,j\le 2m}$.  Second, when we determine the assignment to the $(2m-1)$-th and $2m$-th subjects, we have the record of the assignments to the first $(2m-2)$ subjects, although we cannot update them. Therefore, given the submatrix $A^{(2m)}$ and $T_{1}, \dots T_{2m-2}$, we need to determine $T_{2m-1}$ and $T_{2m}$ to reduce the \emph{imbalance measurement}, defined as
	\begin{align}\label{eq:imbalance:measurement}
	I_{2m} = \| A^{(2m)}({\bf 1}_{2m}-2T_{1:2m})\|,
	\end{align}
	where ${\bf 1}_{2m}\in \mathbb R^{2m}$ with all entries equal to 1 and $T_{1:2m}$ consists of the first $2m$ entries of $T$. To reduce the imbalance measurement, we propose the following procedure:
	\begin{itemize}
		\item[1.] The first two subjects are randomly assigned to different treatments.
		\item[2.] Suppose $2m-2$ patients have been assigned to treatments, we define the imbalance measurement when $(T_{2m-1}, T_{2m}) = (0,1)$
		\begin{align*}
		I_{2m}^{(0, 1)} = \|A^{(2m)}({\bf 1}_{2m} - 2(T_{1:(2m-2)}^\top, 0, 1)^\top)\|,
		\end{align*}
		and in the same manner, when $(T_{2m-1}, T_{2m}) = (1,0)$, we have
		\begin{align*}
		I_{2m}^{(1, 0)} = \|A^{(2m)}({\bf 1}_{2m} - 2(T_{1:(2m-2)}^\top, 1, 0)^\top)\|.
		\end{align*}
		\item[3.] We decide $(T_{2m-1}, T_{2m})$ according to the following probabilities,
		\begin{align*}
		\pr((T_{2m-1}, T_{2m}) = (0,1)) =
		\begin{cases}
		b,	&\text{ if } I_{m}^{(0, 1)} < I_{m}^{(1, 0)}; \\
		1-b, &\text{ if } I_{m}^{(0, 1)} > I_{m}^{(1, 0)}; \\
		0.5, &\text{ otherwise. }
		\end{cases}
		\end{align*}
		Here $b \in (1/2, 1]$ is a fixed biasing probability.
		\item[4.] We repeat steps 2 and 3 until $2m\ge n-1$. If $2m=n-1$, we arbitrarily assign a treatment to subject $n$.
	\end{itemize}
	
	The general idea of this procedure can be summarized as follows. In each stage, we consider two possible assignments to $(T_{2m-1}, T_{2m})$ and compute which assignment minimizes the imbalance measurement. In pairwise sequential randomization, the assignments are either $(T_{2m-1}, T_{2m})=(0,1)$, or $(1,0)$. We use the assignment that results in the smallest imbalance measurement with the biasing probability $b\in (1/2, 1]$. It is clear that letting $b=1$ would reduce the expected imbalance measurement as far as possible, but we allow randomness in the procedure for several practical reasons. We further discuss this biasing probability in Remark~\ref{rem:biasing:probability}. Notably, the proposed procedure does not require any information on subjects joining the experiment in the future. To be more specific, the choice of treatment for subjects $2m-1$ and $2m$ depends only on their connection with previous subjects and the current imbalance measurement. {The procedure can be applied to the case when $n$ is odd, as long as we assign a random treatment to the last subject.} If $b$ is a constant greater than $1/2$, the adaptive procedure can significantly reduce the imbalance measurement under mild assumptions on the network.
	
	\begin{remark}[Biased coin design]\label{rem:biasing:probability}
		Suppose we let $b=1$ in our proposed procedure, then each pair of assignments in the procedure reduces the imbalance measurement as far as possible, and treatment allocation is completely determined by the network. However, deterministic treatment assignment is not desirable from the standpoint of (un)predictability and the principle of randomness~\cite{lewis1999statistical}, so an appropriate allocation probability $\in(1/2,1)$ should be selected. The idea of biased coin design is introduced in~\cite{efron1971forcing} for balancing the total number of different treatments. For the purpose of balancing prognostic factors between treatment groups, the authors of \cite{hagino2004statistical} suggest an allocation probability between 0.70 and 0.95 according to the sample size. In~\cite{toorawa2009use}, the authors simulate the effects of allocation probability. In this paper, we assume that $b$ can be any constant greater than 0.5 and no more than 1.
	\end{remark}
	
	\begin{remark}[Binary Integer Programming]
		Suppose the whole network is observed, the goal of reducing the imbalance measurement $I = \|A(\ones-2T)\|$ with unbiased treatment groups is equivalent to the following optimization problem:
		\begin{align*}
		\min_{x\in\{-1,1\},\  \ones^\top x = 0} \|Ax\| =
		\min_{x\in\{-1,1\},\  \ones^\top x = 0} x^\top H x,
		\end{align*}
		where $H = A^\top A=A^2$. It is not difficult to observe that $H_{ij}$ counts the number of common neighbors of node $i$ and $j$ in the adjacency matrix $A$. The constrained $1^\top x=0$ can be converted to a penalty function:
		\begin{align*}
		\min_{x\in\{-1,1\}}  x^\top Hx + \lambda(\ones^\top x)^2
		=\min_{x\in\{-1,1\}}  x^\top (H+\lambda \ones \ones^\top )x.
		\end{align*}
		This formulation is summarized as an unconstrained binary programming problem (UBQP) in ~\cite{kochenberger2014unconstrained}. The authors of that survey also mention that the UBQP is an NP-hard problem, whose proof is provided in~\cite{pardalos1992complexity}, except for some special cases with very strong assumptions on $H$~\cite{picard1976maximal, barahona1986solvable, pardalos1991graph}. $H$ in these special cases is restricted to be an adjacency matrix with certain regularization conditions, so their results cannot apply to our case $H=A^2$. In the general case, heuristic methods such as the continuous approach~\cite{pardalos2006continuous, pan2008global}, tabu search algorithms ~\cite{lu2011neighborhood, wang2013probabilistic}, and semi-definite relaxation~\cite{wang2013fast} have been proposed for finding inexact but high-quality solutions. However, it is worth noting that the setting we consider here is very different from a UBQP problem. We have to determine $x_i$ when only the upper-left $i\times i$ submatrix of $A$ is observed.
	\end{remark}
	
	\section{Theoretical Properties of the Proposed Design}\label{sec:theory}
	
	In this section, we study the asymptotic property of the imbalance measurement quantity of~\eqref{eq:imbalance:measurement} under the following  stochastic assumption on the symmetric adjacency matrix $A$. We assume for some $p\in(0,1)$,
	\begin{align}\label{eq:erdos:renyi:assumption}
	A-I \sim G(n,p), \quad \text{where } G(n,p) \text{ represents the Erd\H os-R\' enyi random graph model. }
	\end{align}
	In other words, on the diagonal of $A$, we have determinant entries $A_{ii} = 1$ for $i \in[n]$, and
	\begin{align*}
	A_{ij} = A_{ji}\sim \text{Bernoulli}(p) \text{ independently for } 1\le i<j\le n.
	\end{align*}
	In the graph sense, the Erd\H os-R\' enyi random graph model indicates that an edge between distinct nodes exists with probability $p$~\cite{erdHos1960evolution}. Under this assumption on $A$, we aim to analyze the asymptotic behavior of the imbalance measurement $I_{2m} = \| A^{(2m)}({\bf 1}_{2m}-2T_{1:2m})\|$ defined in \eqref{eq:imbalance:measurement}.  Let us also define the state after the $m$-th iteration of the procedure:
	\begin{align*}
	S_{2m} = A_{2m}({\bf 1}_{2m}-2T_{1:2m})
	\end{align*}
	so that $I_{2m} = \|S_{2m}\|$. {For convenience of notation, we let
		\begin{align}\label{eq:odd:I}
		I_{2m+1} = I_{2m} \quad \text{ for } m\in\mathbb N,
		\end{align}
		so the imbalance measurement $I_i$ can be defined for all positive integers $i$. }
	Suppose $A$ were not symmetric, i.e., $a_{ij}$ and $a_{ji}$ were i.i.d., then $\{S_i\}_{i \in\mathbb N}$ would be a \emph{time-varying Markov chain}, where the randomness comes from entrywise Bernoulli distribution and random assignments in step 3 of the procedure. In the symmetric case, we still approximately have the following Markov property:
	\begin{align*}
	\pr(S_i = x|S_1, \dots, S_{i-1})\approx  \pr(S_i = x|S_{i-1})
	\end{align*}
	We will show that the imbalance measurement $I_n$ is significantly reduced compared with random design if we apply our proposed procedure. \par
	A random design indicates that we assign $(T_{2m-1}, T_{2m}) = (0,1)$ or $(1,0)$ with probability $1/2$. In other words, we implement step 3 of our proposed design with $b=1/2$. We denote the resulting assignments by the vector $T_{\text{random}}$, then for fixed $p\in(0,1)$, we have the following theorem about random assignment.

	\begin{theorem}\label{thm:random}
		Suppose the $n\times n$ network follows the Erd\H os-R\' enyi random graph model in~\eqref{eq:erdos:renyi:assumption} with Bernoulli parameter $p$, then using random assignment, the imbalance measurement satisfies the following limit
		\begin{align}\label{eq:random:design:limit}
		\lim_{n\to\infty} \frac {\ex[\|A (1-2T_{\text{random}})\|^2]}{n^2}  = p(1-p).
		\end{align}
	\end{theorem}
	
	The next theorem shows that our proposed design can significantly reduce the imbalance measurement.
	
	\begin{theorem}\label{thm:adaptive}
		Suppose the $n\times n$ network follows the Erd\H os-R\' enyi random graph model in~\eqref{eq:erdos:renyi:assumption} with Bernoulli parameter $p$, then using our proposed design, the imbalance measurement $I_n$ satisfies the following upper bound:
		\begin{align}\label{eq:adaptive:design:limit}
		\lim\sup_{n\to\infty} \frac{\ex[I_n^4]}{n^4}\le p^2(1-p)^2 - \frac 1{8}(2b-1)(2-\sqrt{2}(2b-1))^{3/2} p^{5/2}(1-p)^{5/2}.
		\end{align}
	\end{theorem}
	
	\begin{remark}\label{rem:adaptive}
		Theorem~\ref{thm:adaptive} provides the upper bound of the fourth moments of the imbalance measurement. Because $\ex[I_n^4]\ge \ex[I_n^2]^2$, we immediately obtain an upper bound of the second moment $\ex[I_n^2]$. For fixed $p\in(0,1)$ and $b\in(1/2,1]$,
		\begin{align*}
		\lim\sup_{n\to\infty} \frac{\ex[I_n^2]}{n^2}<p(1-p).
		\end{align*}
		Hence the proposed procedure provides a strictly smaller imbalance measurement than random design in expectation. Suppose $b=1/2$, i.e., $2b-1=0$, then the proposed method is identical to random design. As a result, the second term of~\eqref{eq:adaptive:design:limit} vanishes. Meanwhile, suppose the network is very sparse, that is, $p$ is very small, then the reduction of imbalance measurement by the proposed design is not very great, because $p^{5/2}$ is much smaller than $p^2$.
	\end{remark}

	\section{Discussion on the Gaussian Case}\label{sec:gaussian}
	
	In previous sections, we assumed the network followed the Erd\H os-R\' enyi random graph model. As discussed in Remark~\ref{rem:adaptive},  Theorem~\ref{thm:adaptive} has not shown that the reduction of imbalance measurement is asymptotically smaller than the imbalance measurement itself if $p\to 0$. To discuss whether the reduction is rate optimal, we investigate a weighted adjacency matrix with Gaussian entries. Specifically, in Wigner matrix $A$, we have determinant entries $A_{ii} = 1$ for $i \in[n]$, and
	\begin{align*}
	A_{ij} = A_{ji}\sim \mathcal N(0, \sigma^2) \text{ independently for } 1\le i<j\le n.
	\end{align*}
	In other words, we consider the Gaussian orthogonal ensemble (GOE) instead of the Erd\H os-R\' enyi random graph model. This assumption corresponds to the following scenario. We observe a weighted network in which the weights can be either positive or negative. Under assumption~\eqref{eq:model:assump}, the observation $X_i$ is still well defined. In this case, the unknown covariate $Z_j$ can affect the $i$-th observation $X_i$ positively or negatively, depending on the weight $A_{ij}$. Under this assumption, the proposed procedure in Section~\ref{sec:model} is still valid. If we adopt the definition of imbalance measurement $I_n$, we have the following asymptotic upper bound.
	
	\begin{theorem}\label{thm:GOE}
		Suppose the $n\times n$ weighted network follows the GOE with variance $\sigma^2$ where $\sigma$ depends on $n$. Assuming $\sigma=O(1)$ and $n\sigma^2\to\infty$, then using our proposed design, the imbalance measurement $I_n$ satisfies the following upper bound:
		\begin{align*}
		\lim\sup_{n\to\infty} \frac{\ex[I_n^4]}{n^4\sigma^4 }
		\le 1- \frac 14(2b-1)\sqrt{2/\pi}(4-\sqrt {2/\pi}(2b-1))^{3/2}.
		\end{align*}
	\end{theorem}
	
	In the Erd\H os-R\' enyi random graph model, the entrywise variance of the adjacency matrix is $p(1-p)$. This quantity is comparable to $\sigma^2$ in the GOE. When $\sigma\to 0$, the reduction of the imbalance measurement is still significantly large. This is a stronger result than that in the Erd\H os-R\' enyi random graph model. An essential technical reason is the lower bound of $\ex[|x^\top Y|]$ for fixed subject vector $x$ and centered random vector $Y\in\mathbb R^m$. If we assume only $Y$ is a sub-Gaussian vector, then for general $x$, we obtain the best lower bound by Khinchin-Kahane inequality, see Lemma~\ref{lem:kahane}. If we further assume $Y_i\sim\mathcal N(0,\sigma^2)$ independently, then $|x^\top Y|$ is a folded normal random variable and $\ex[|x^\top Y|] = \sigma\sqrt{n/\pi}$ for all subject vectors $x$. It is still an open problem whether the term $p^{5/2}$ can be improved to $p^2$. An empirical comparison of these two cases can be found in the next section.

	\section{Experiments} \label{sec:exp}
	
	In this section, we empirically study the behavior of imbalance measurement in \eqref{eq:imbalance:measurement}. The experiments demonstrate that our proposed algorithm improves the estimation of treatment effects for both simulated data and real network data.
	
	\subsection{Experiments on Simulated Network Data}
	
	\begin{figure}[h!]
		\includegraphics[width=2.5in]{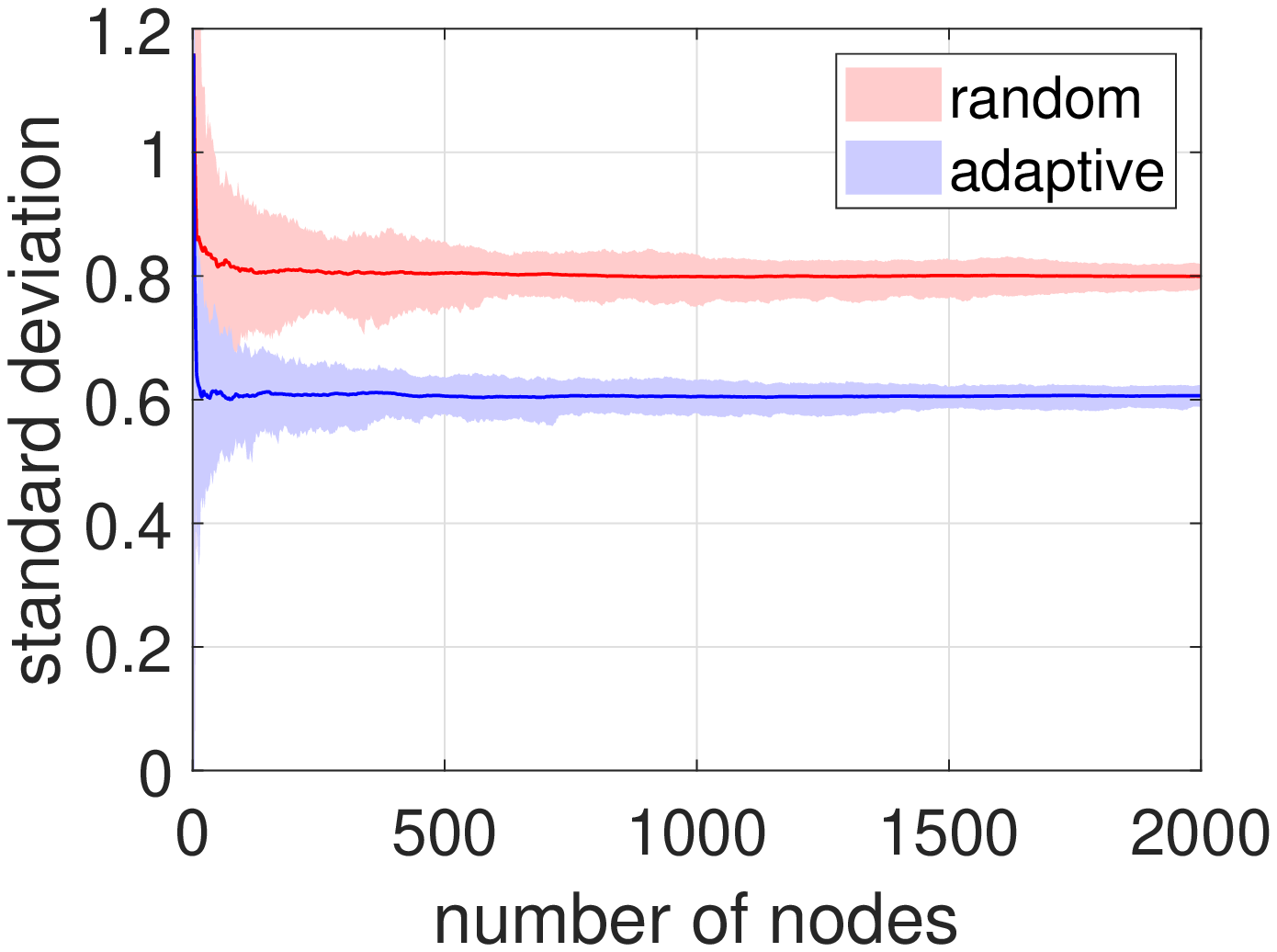}
		\includegraphics[width=2.5in]{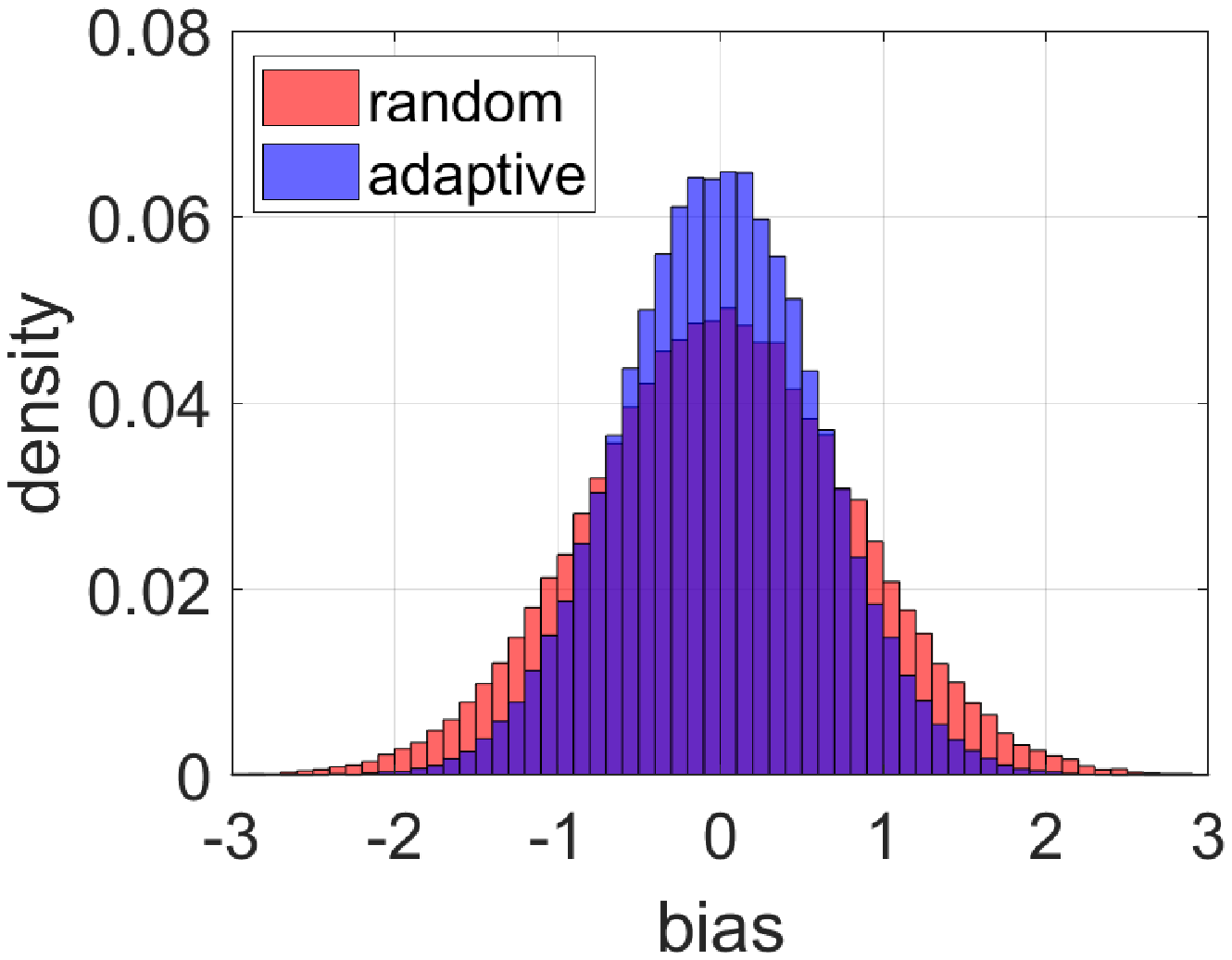}
		\caption{Left: the standard deviation of $W$ in Section~\ref{sec:model} for different $n$. Right: the histogram of $W$ when $\mu_0=\mu_1$ and $\sigma_\varepsilon=1$ }
		\label{fig:one}
	\end{figure}
	The plots in Figure~\ref{fig:one} show the result of the Erd\H os-R\' enyi random graph model in~\eqref{eq:erdos:renyi:assumption}. We fix $p=0.2$ and simulate different sizes of random graphs. We consider random assignment and our proposed adaptive design algorithm with $b=0.95$. On the left plot, the shaded region is the $95\%$ confidence interval for 100 iterations. All other plots with shaded regions in this section have confidence intervals with the same confidence coefficient. The plot of random assignment shows that the imbalance measurement concentrates around 0.8. This coincides with the theoretical limit $\sqrt{4p(1-p)}=0.8$ suggested by Theorem~\ref{thm:random}. Applying the proposed algorithm, the imbalance measurement decreases to approximately 0.6. The right plot shows the bias of estimation of $\mu_0-\mu_1$ when $n=100$. This experiment is repeated 20000 times.

		\begin{figure}[h]
		\includegraphics[width=2.5in]{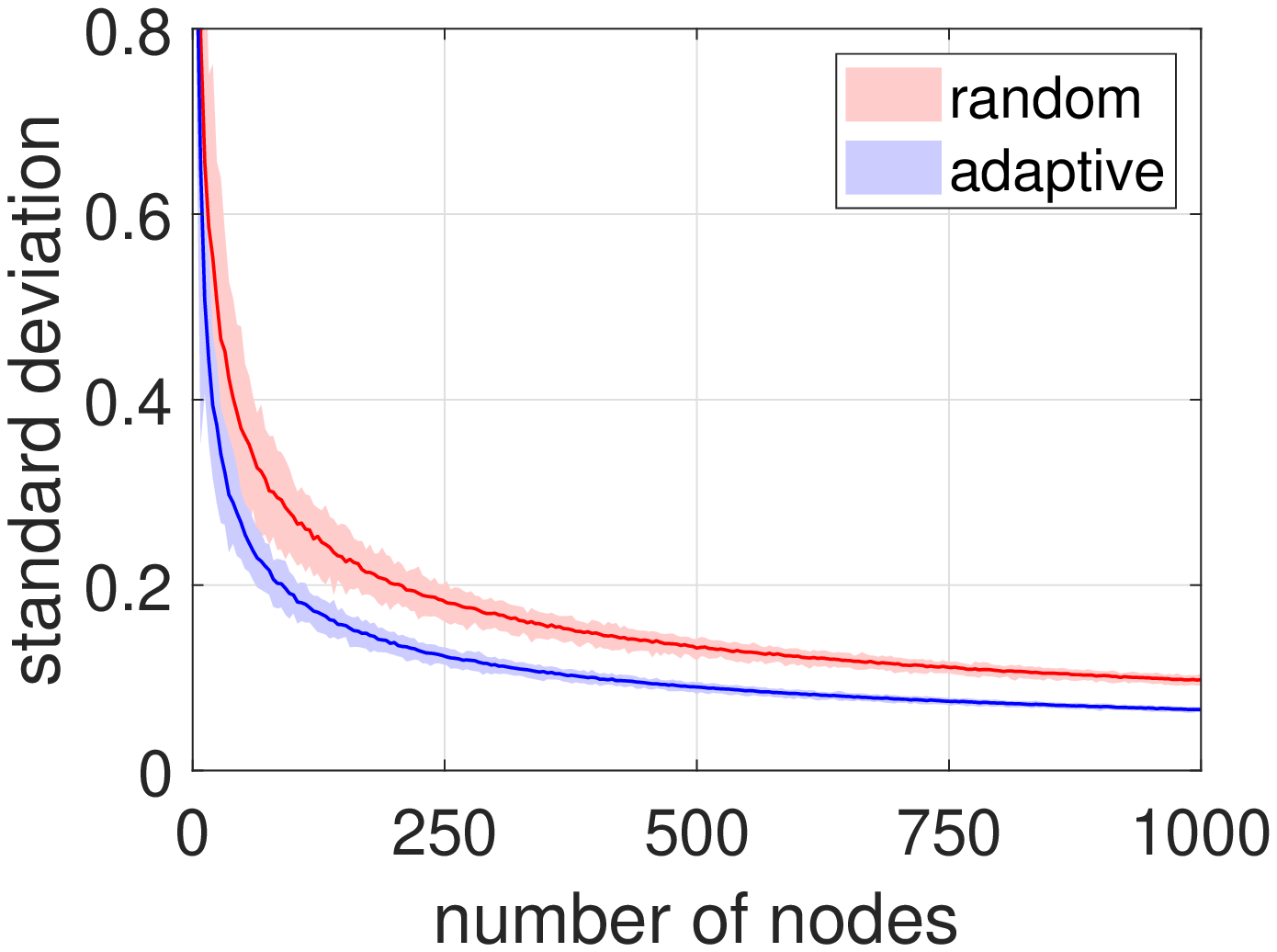} \includegraphics[width=2.5in]{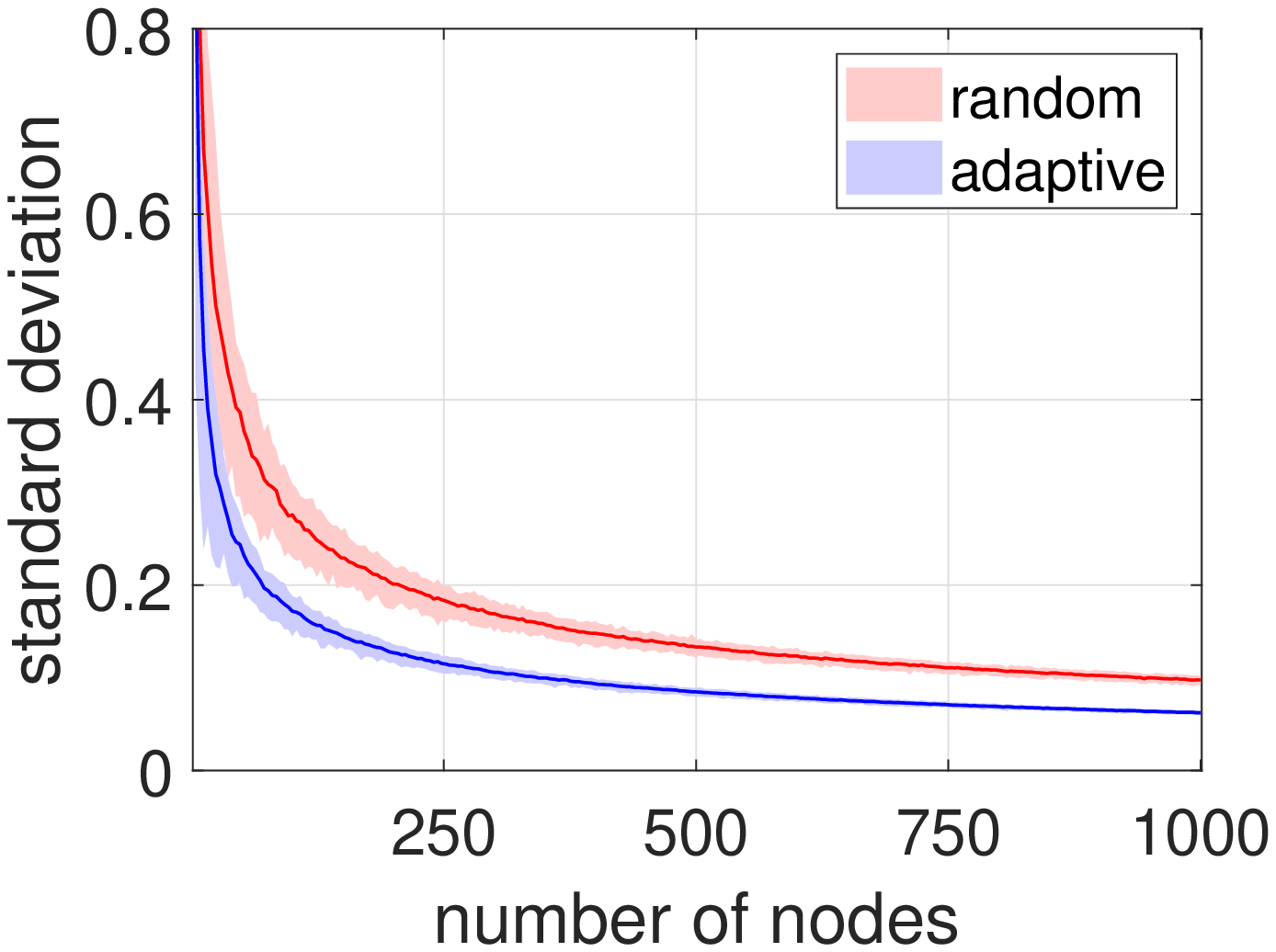}
		\caption{Left: the result when the graph is sparse. Right: the simulation on the Gaussian orthogonal ensemble.}
		\label{fig:two}
	\end{figure}

	The left plot in Figure~\ref{fig:two} considers the sparse Erd\H os-R\' enyi random graph model. For an $n\times n$ network, we consider the density regime $\frac{\log n}{n}$. In particular, we generate random networks $G\big(n, \frac{\log n}{5n}\big)$. The shaded region is the interquartile range over 100 iterations. The imbalance measurement of random design monotonically decreases because its maximum expectation is $\sqrt{4p(1-p)}$, which converges to 0 as the network becomes more sparse. The right plot in Figure~\ref{fig:two} considers the GOE instead. The entrywise variance remains the same. In other words, for the network with $n$ nodes, $p_n=\frac{\log n}{n}$, the corresponding variance of the GOE is $\sigma_n^2 = p_n(1-p_n)$. The results of this simulation show that the imbalance measures have very similar asymptotic behavior for both models.
	
		\begin{figure}[h]
		\mbox{
			\includegraphics[width=0.33\textwidth]{fig/ER_p2.eps}
			\includegraphics[width=0.33\textwidth]{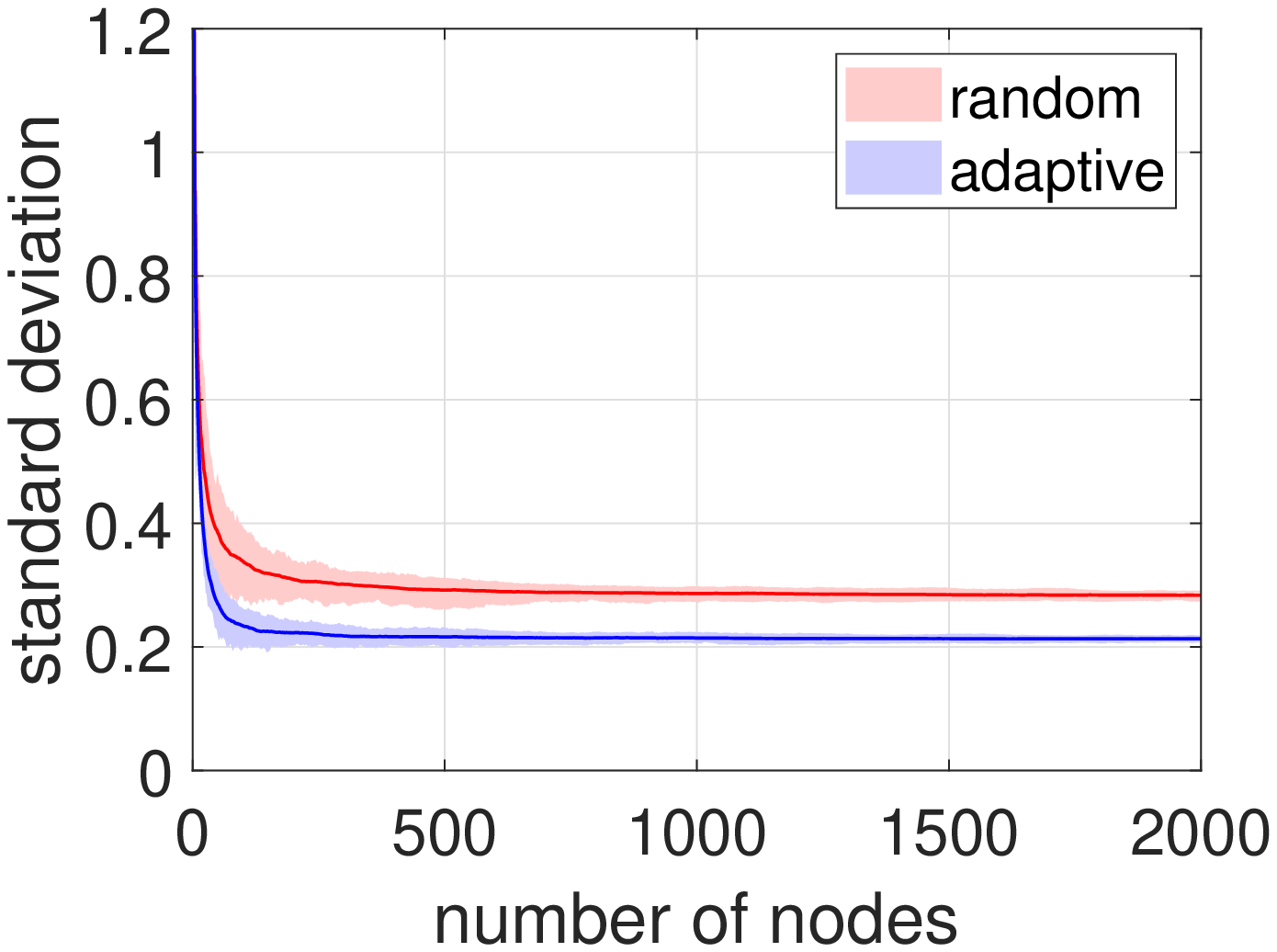}
			\includegraphics[width=0.33\textwidth]{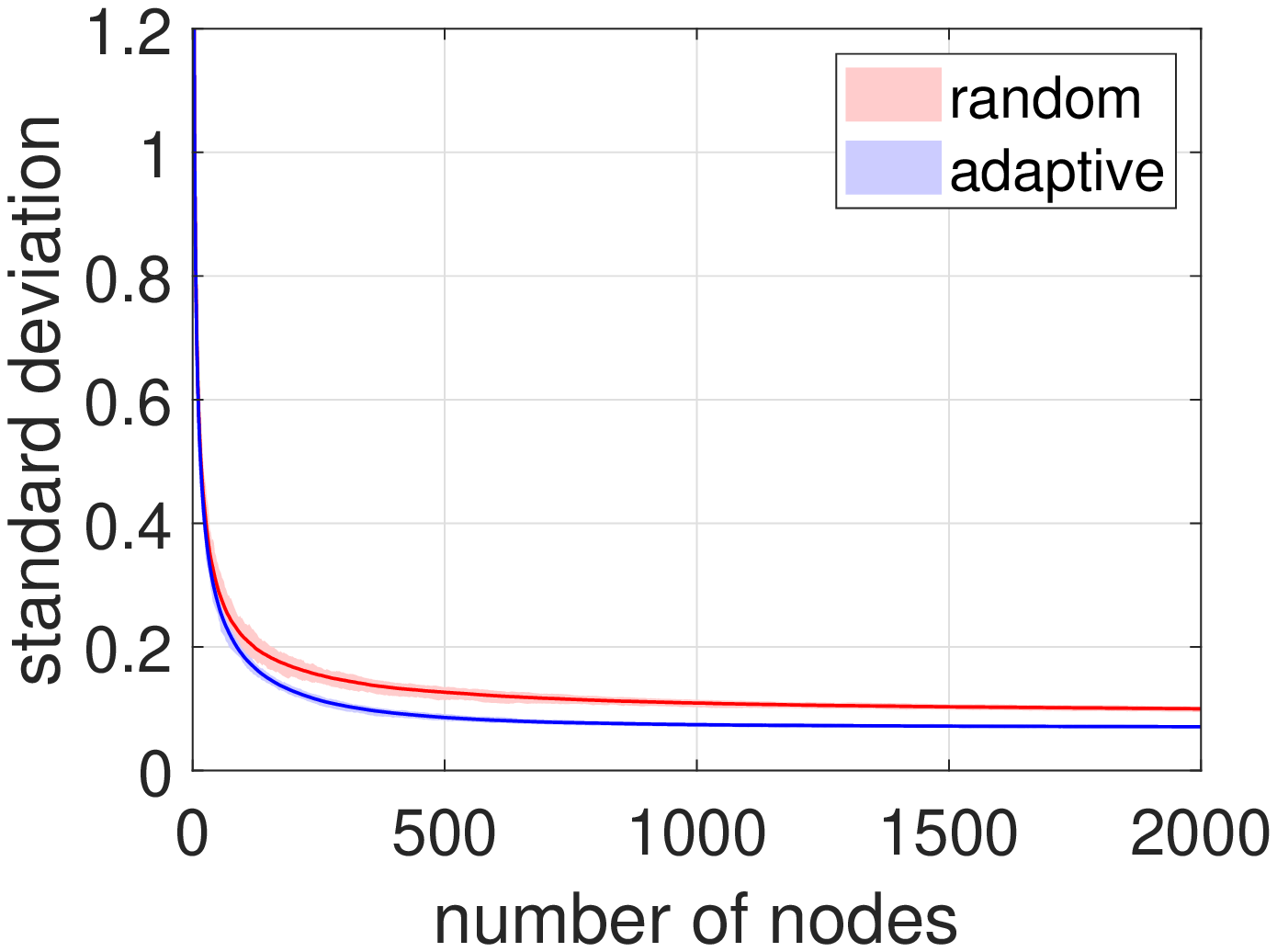}
		}
		\caption{Comparison of the Erd\H os-R\' enyi random graph model with different density. }
		\label{fig:three}
	\end{figure}

	The plots in Figure~\ref{fig:three} compare the performance of our proposed method on the Erd\H os-R\' enyi random graph model. In the plots, we let $p=0.2, 0.02$, and $0.002$ and plot the imbalance measurement on different sizes of random graphs. The result of this experiment is identical to that shown in the left plot in Figure~\ref{fig:one}, but in different densities.
	
	\newpage

	\begin{figure}[h]
		\mbox{
			\includegraphics[width=0.33\textwidth]{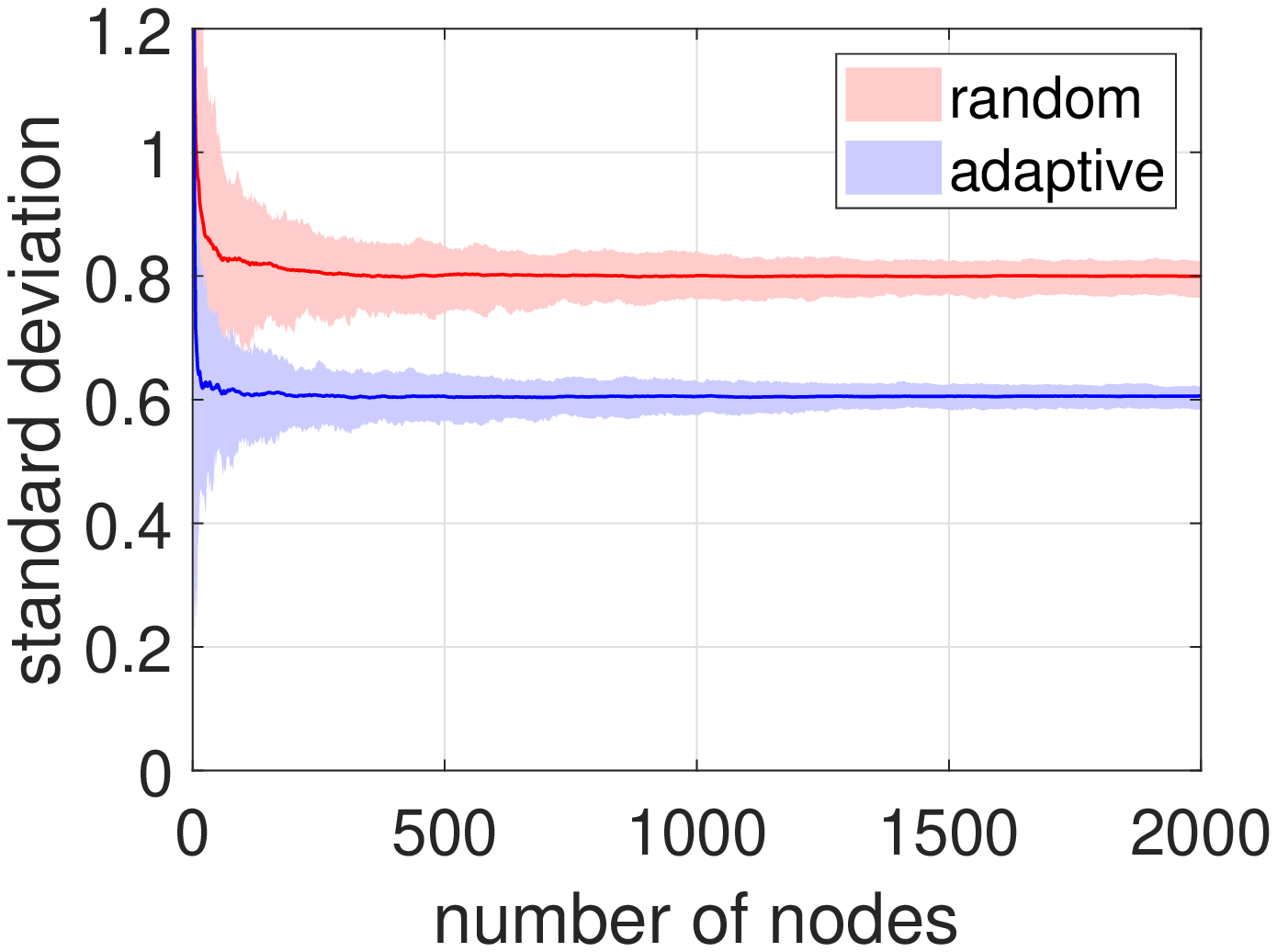}
			\includegraphics[width=0.33\textwidth]{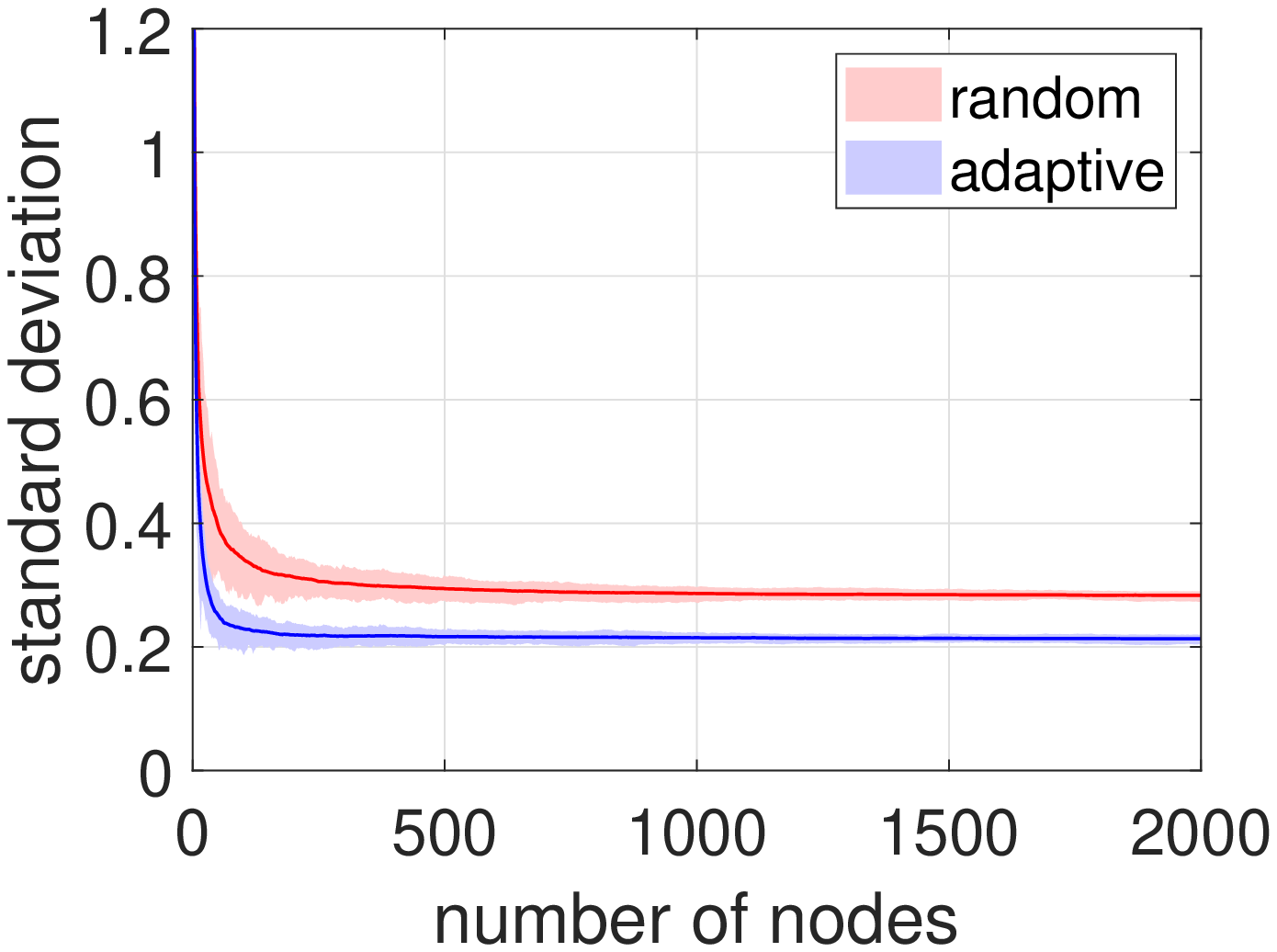}
			\includegraphics[width=0.33\textwidth]{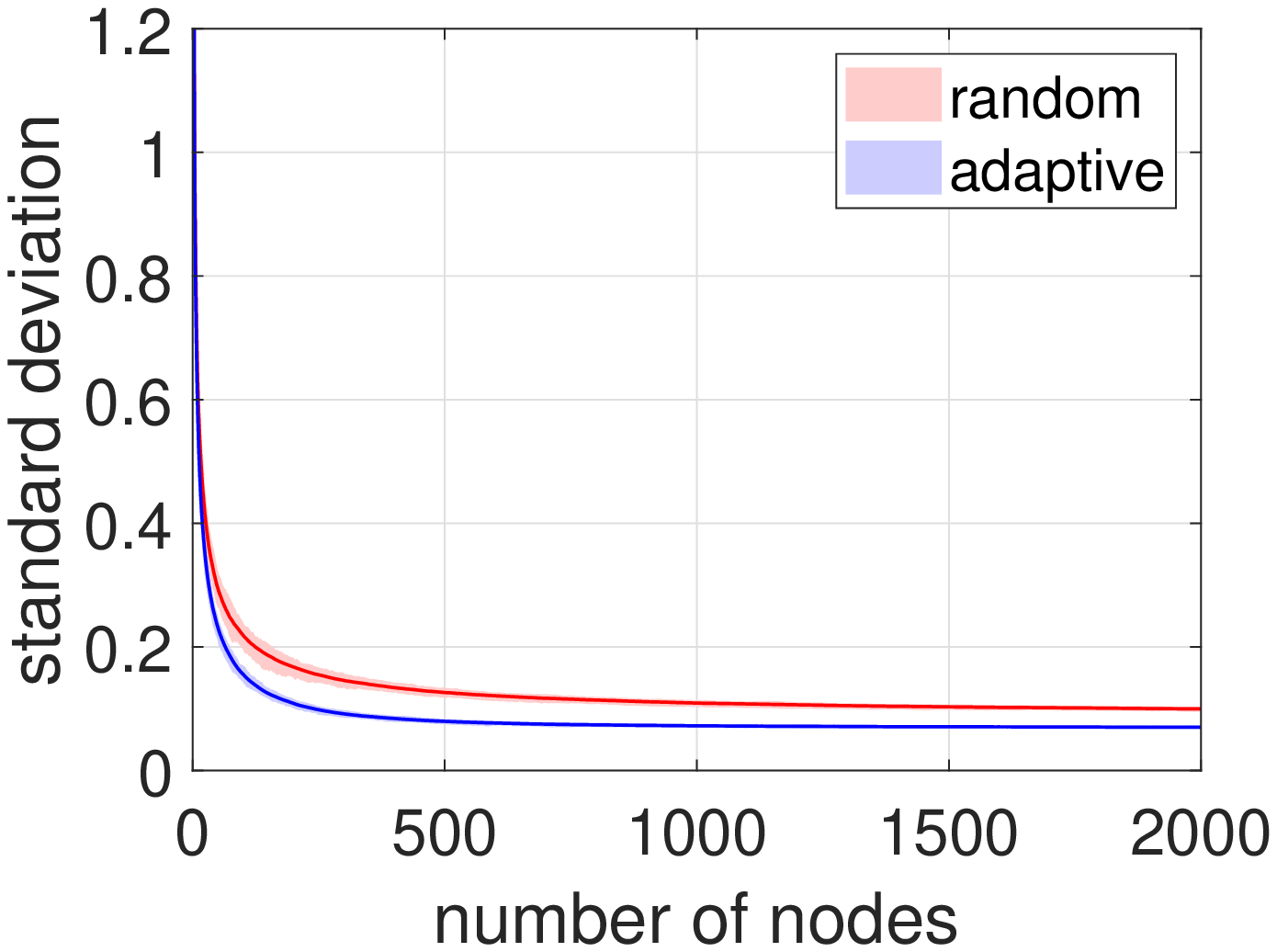}
		}
		\caption{Comparison of Gaussian orthogonal ensemble with different variance. }
		\label{fig:four}
	\end{figure}

	In Figure~\ref{fig:four}, we consider the GOE with different levels of variance. For the left, center, and right plots, we let $\sigma^2=p(1-p)$ for $p=0.2, 0.02, 0.002$ respectively, so that the entrywise variance is the same as that in the Erd\H os-R\' enyi random graph model. As we can see in the plots, the empirical performances are very similar for these two models.
	
	\begin{figure}[h]
		\mbox{
			\includegraphics[width=0.33\textwidth]{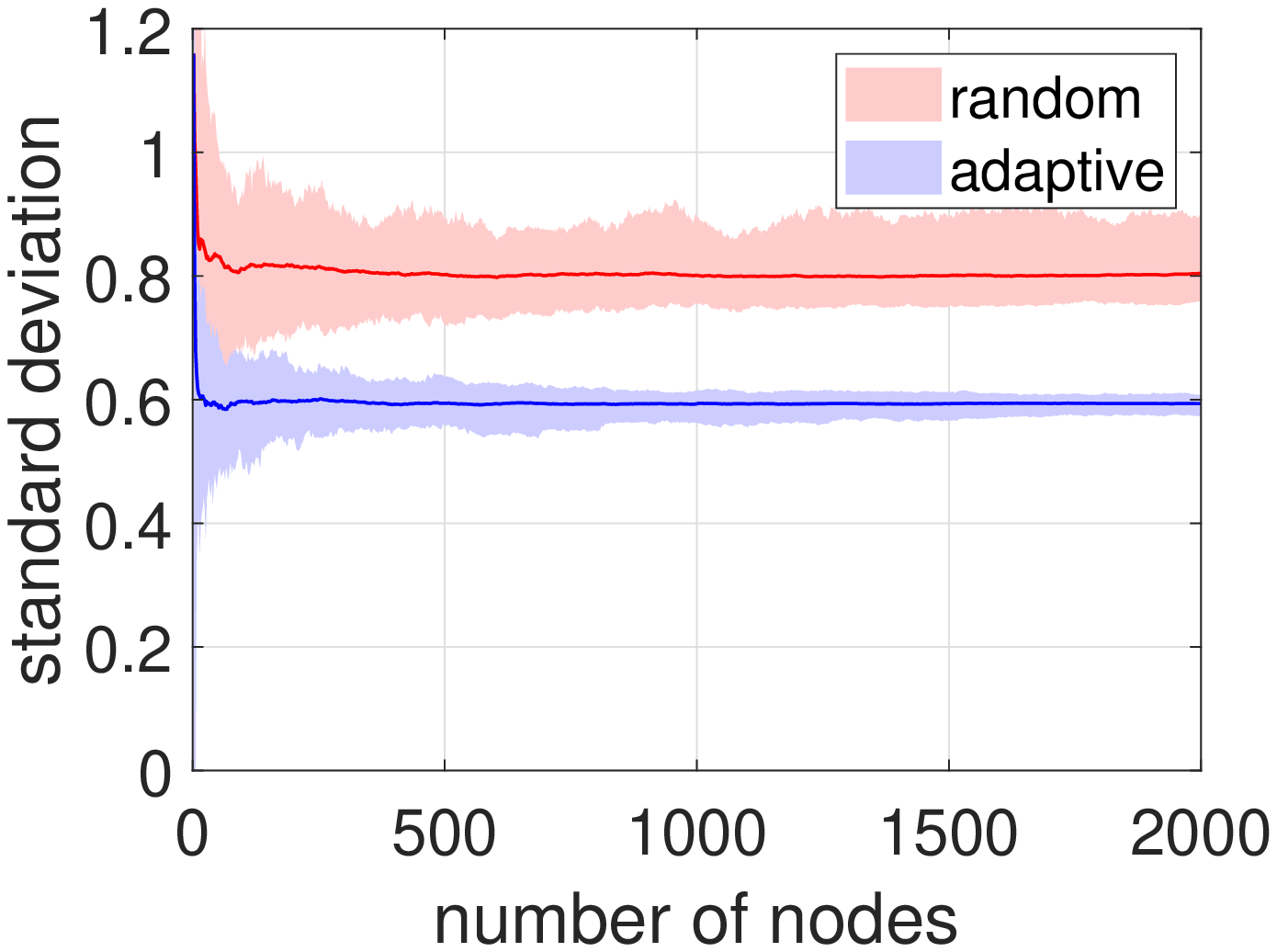}
			\includegraphics[width=0.33\textwidth]{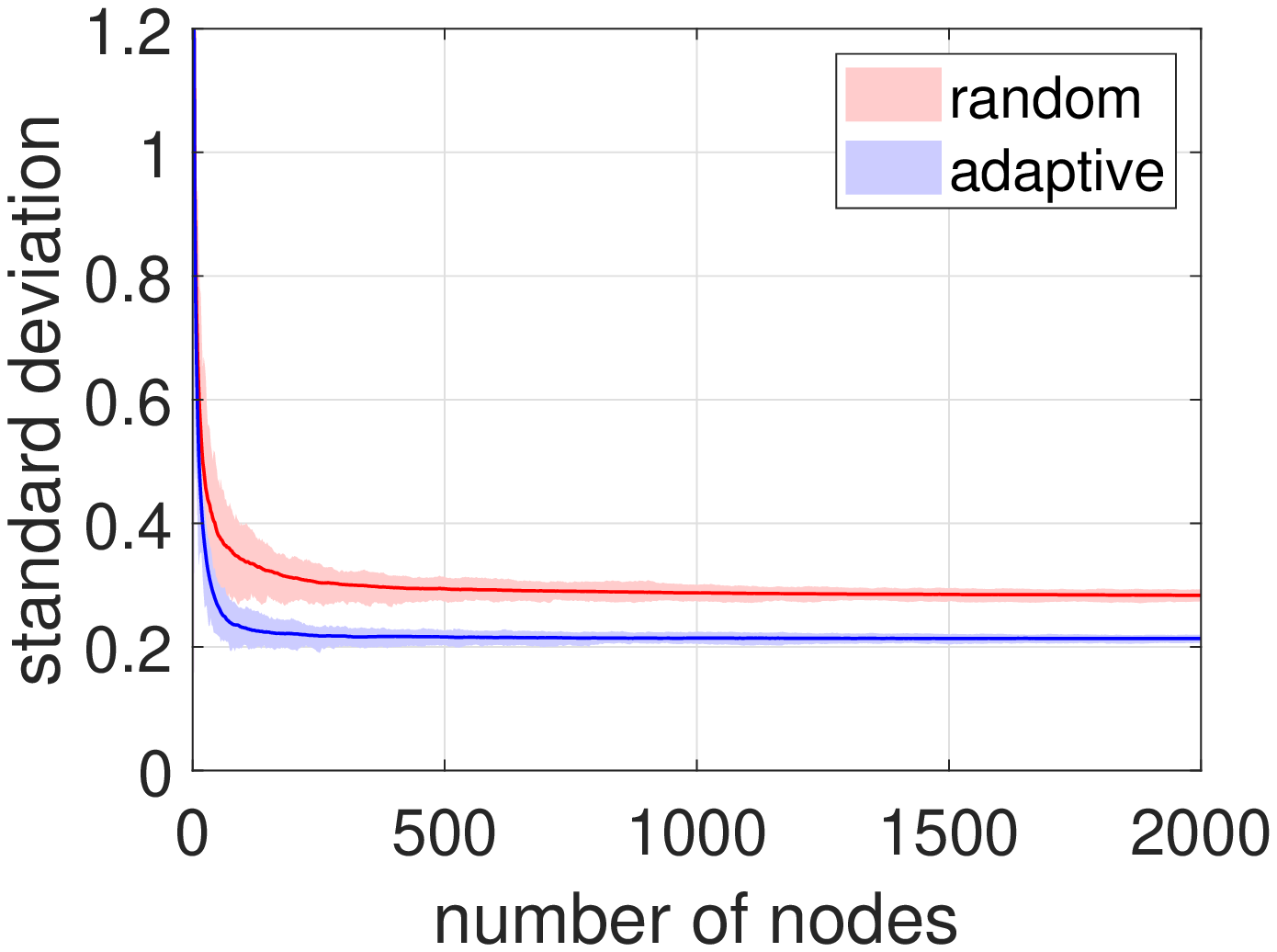}
			\includegraphics[width=0.33\textwidth]{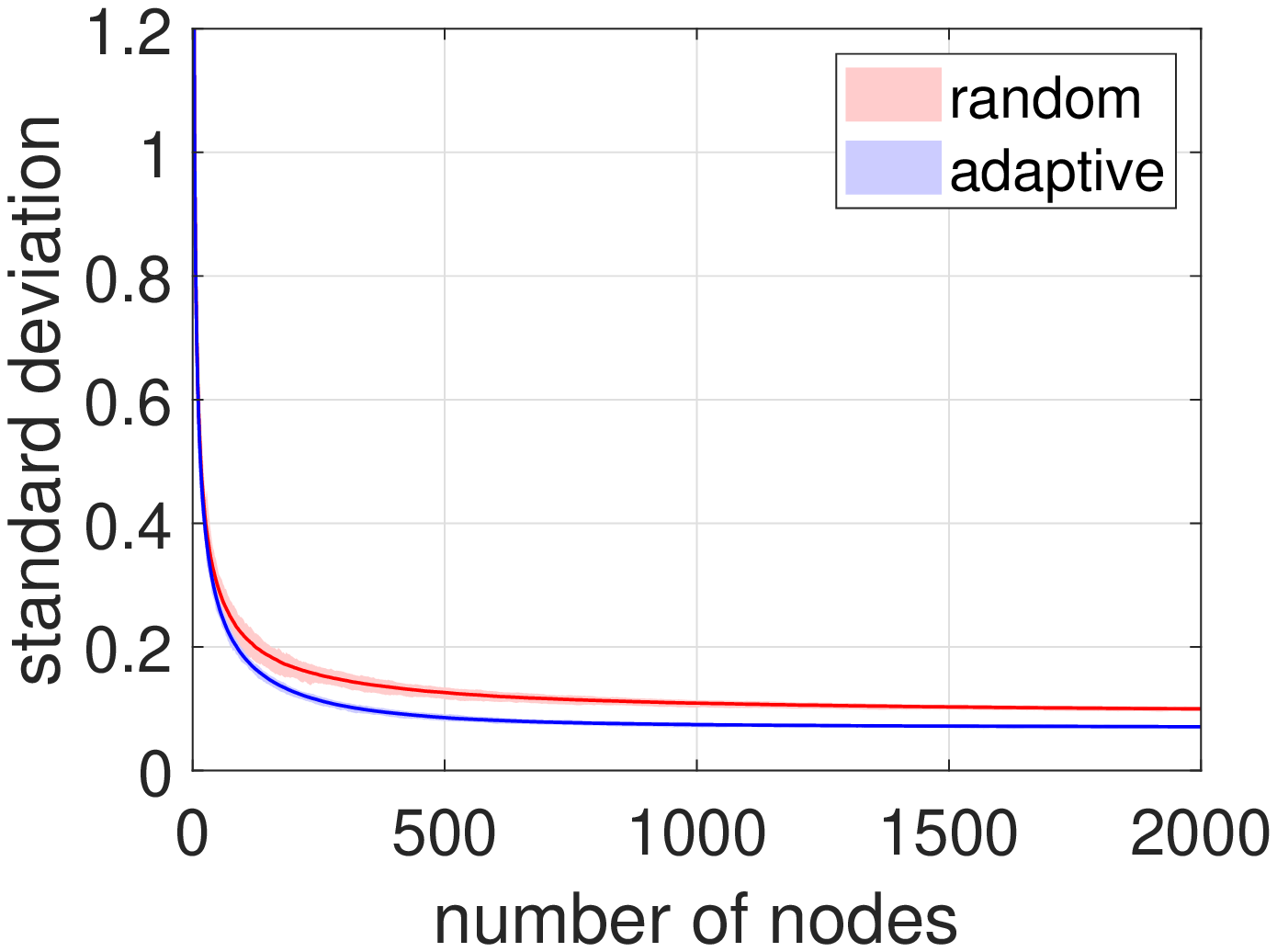}
		}
		\caption{Comparison of stochastic block models with different density. }
		\label{fig:five}
	\end{figure}

	Figure~\ref{fig:five} considers the stochastic block model. The subjects are randomly divided into two groups. In our setting, if two subjects belong to the same group, then the probability of connection between them is $p_1$, and the between-group probability is $p_2$. The plots consider $p_1=0.3, 0.03, 0.003$ and $p_2=0.1, 0.01, 0.001$ from left to right respectively. The overall density is the same as the previous experiment on the Erd\H os-R\' enyi random graph in Figure~\ref{fig:three}, and the empirical performances of the two experiments are very similar. However, on the left plot, we can observe that the confidence intervals are wider than those in Figure~\ref{fig:three} and Figure~\ref{fig:four}.

	\subsection{Experiments on Real Network Data}
	
	We implement our proposed algorithm on 11 real undirected network datasets from SNAP~\cite{snapnets}. For each network, we randomly sample a subnetwork with 10000 nodes, then we apply both adaptive design and random design to the subnetwork. We compare the imbalance measurement in \eqref{eq:imbalance:measurement} in each dataset. Because theoretical analysis shows that the network density plays an important role in imbalance measurement, the densities are recorded in the last column. In our model, we always assume the existence of self loops. When we compute the density of the subnetwork, we consider only the connections between different nodes. For example, the density of the network \verb|com-youtube.ungraph| is approximately 0.00076. In other words, the average degree of the graph is approximately 8.6, including self loops.
	
	\begin{table}[h]
		\caption{~\label{tab:dataset} Comparison of adaptive random design applied to real network data from SNAP}
		\begin{tabular}{ccccc}
			\hline
			Dataset & Adaptive & Random & Reduction & Density\\
			\hline
			\verb|Email-Enron| & 212.4147 & 346.6583  & 39\% &  $11.7178\times 10^{-4}$\\
			\verb|com-youtube.ungraph| &172.6210 & 269.6257 & 36\% & $7.6078\times 10^{-4}$\\
			\verb|HR_edges| & 143.2271 & 212.1933 &  33\% & $3.4652\times 10^{-4}$ \\
			\verb|HU_edges| & 117.8813 & 174.4821 & 32\% &  $1.9725\times 10^{-4}$ \\
			\verb|RO_edges| & 106.7895 & 157.6959 & 32\% & $1.4708\times 10^{-4}$  \\
			\verb|CA-GrQc|  &  99.2774  &  117.0982  & 15\% & $0.3776\times 10^{-4}$\\
			\verb|CA-HepPh|  &  100.3494 & 113.8508  & 12\% & $0.3766\times 10^{-4}$\\
			\verb|CA-AstroPh|  &  100.2098  &  112.5256  & 11\% & $0.2304\times 10^{-4}$\\
			\verb|CA-CondMat|  &  97.8877  &  107.7311  & 9\% & $0.1720\times 10^{-4}$\\
			\verb|CA-HepTh|  &  99.1060  &  104.9667  & 6\% & $  0.1200\times 10^{-4}$\\
			\hline
		\end{tabular}
		\label{table:comparison}
	\end{table}
	
	The edges in these networks might have different meanings. We now explain how our model and algorithm are applied on the network \verb|com-youtube.ungraph|. The nodes of this network represent users on YouTube. Two nodes are connected in the network if they are friends on YouTube. Under the assumption of a network-correlated outcome, friends share common unknown factors that affect observations. To reduce the effect of a factor, we should propose treatment allocation such that friends sharing the corresponding factor are divided into two treatment groups. Suppose we apply our proposed adaptive design on this network with $b=0.85$ (see Remark~\ref{rem:biasing:probability} for the choice of $b$), then the imbalance measurement is reduced by $36\%$.

	We also implement the proposed adaptive design algorithm in the other real networks. The results are presented in Table~\ref{table:comparison}. We observe that the percentage of imbalance measurement reduction depends on the density of the network. For instance, the network of \verb|CA-HepTh| has a low density of $0.1200\times 10^{-4}$, hence, our method can reduce the imbalance measurement by only $6\%$. According to Remark~\ref{rem:adaptive}, there is no evidence that our proposed method can reduce the imbalance measurement significantly if the network is very sparse.
	
		\begin{figure}[h]
		\includegraphics[width=2.5in]{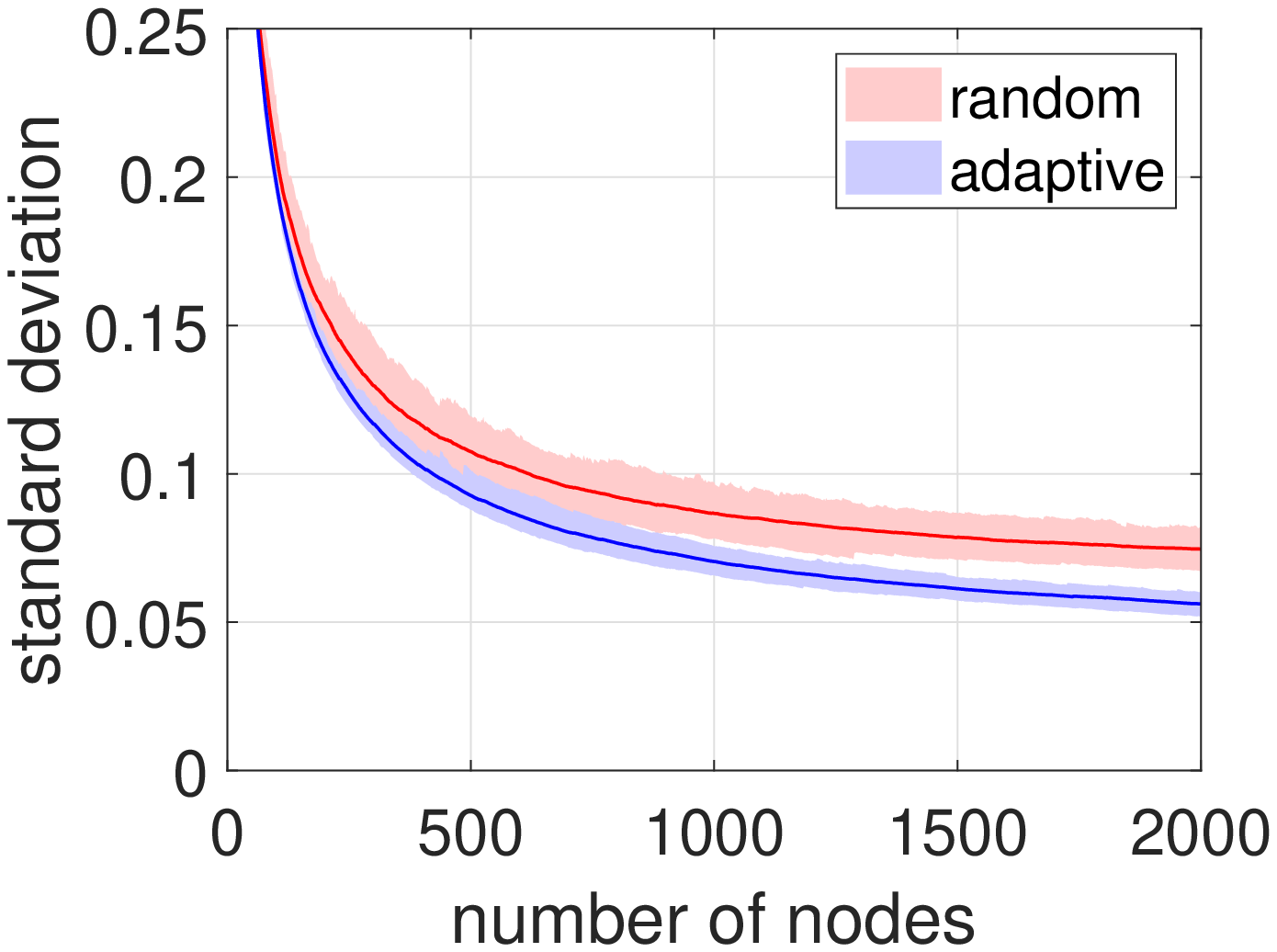}
		\includegraphics[width=2.5in]{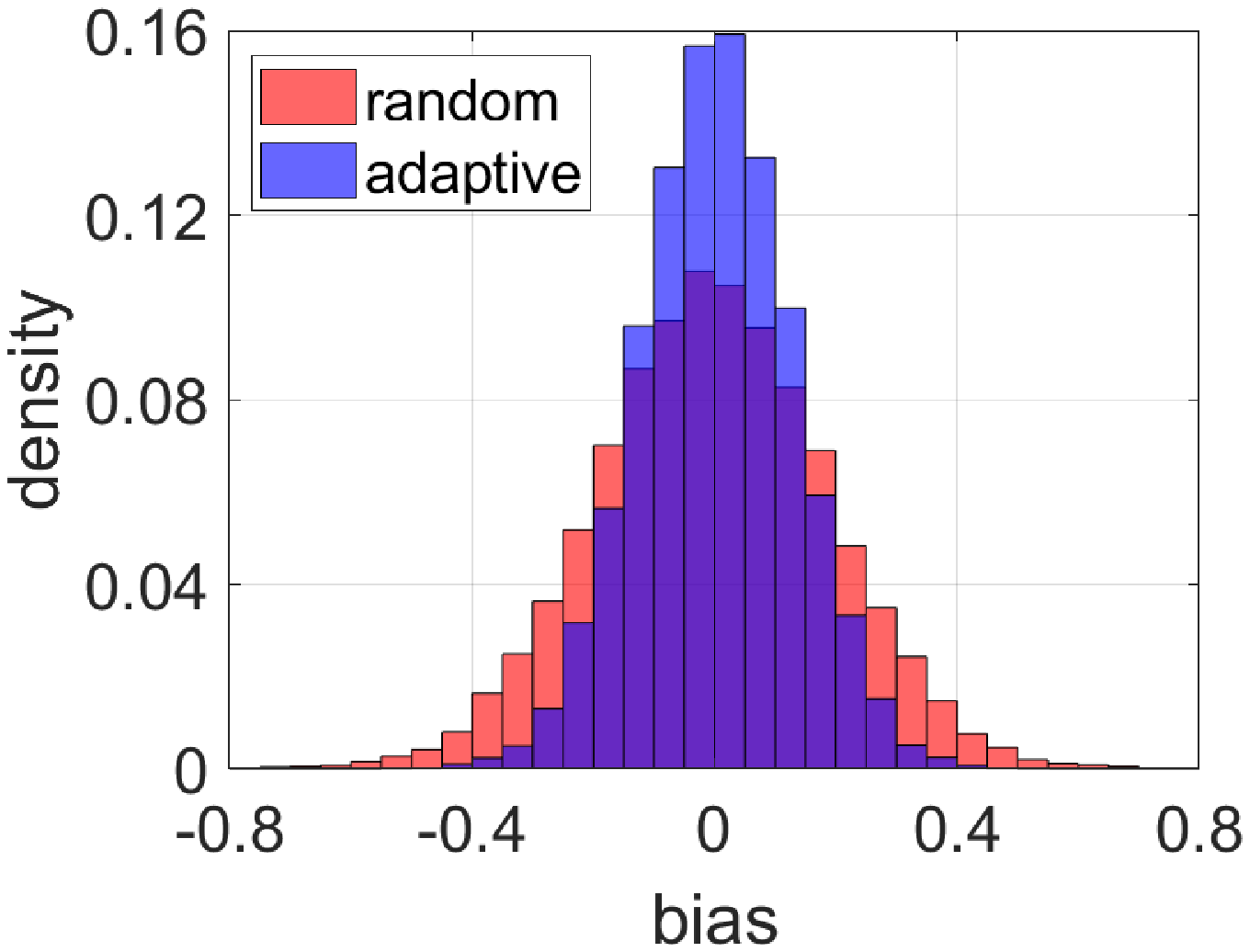}
		\caption{These plots repeat the experiments in Figure~\ref{fig:one} on the real dataset from YouTube.}
		\label{fig:youtube}
	\end{figure}

	In Figure~\ref{fig:youtube}, we implement both random and adaptive designs on different sizes of the network \verb|com-youtube.ungraph|. To repeat the experiments, we keep the random graphs generated in the previous experiment. In this real dataset, to obtain the confidence interval, we repeatedly sample subgraphs from the network and apply the proposed algorithm on each subgraph. The empirical results again show that our proposed method significantly reduces the imbalance measurement and improves the accuracy of estimation of treatment effects.

\newpage

	\section{Conclusion}\label{sec:conclusion}
	
	In this paper, we consider the problem of estimating treatment effects under the assumption that the outcomes are network-correlated. We propose an adaptive randomization procedure to reduce the variance of the estimation. The algorithm assigns different treatments to each pair of subjects sequentially. The biased coin design enforces the assignments, with the result that a smaller imbalance measurement will be chosen with higher probability. For theoretical analysis, we assume the network is generated by the Erd\H os-R\' enyi random graph model. As the number of subjects increases, the states of the Markov process have different dimensions as time progresses. We provide a novel mathematical proof that our adaptive randomization algorithm significantly reduces the imbalance measurement.  Our empirical results also show that this proposed algorithm reduces the variance of the unbiased estimator in both simulated and real data.
	
	The new procedure can still be generalized in several ways. To guarantee a balanced treatment allocation, we consider pairwise sequential design, which determines treatments to two subjects simultaneously. Conventional adaptive design~\cite{hu2012asymptotic} is still applicable in network data. The empirical results can be expected to be similar to the proposed method, but the theoretical analysis will be different, and this is an interesting topic for future work. In the methods described so far, uniform weights are assigned to subjects. If different weights were allowed for different subjects, it might be possible to further reduce the variance of estimation. If a subject has a high degree in the network, its outcome is affected by many other subjects. As a result, the outcome of this node has high variance. If we can reduce the weight of such nodes, the performance of the algorithm will be further improved. If we assume there is interference between subjects, i.e., the outcome of a certain subject might be affected by the treatment of its neighbors, then the analysis in this paper is no longer strictly applicable. However, the tools introduced in this paper could still powerfully reduce the variance under such an interference assumption. As long as we can define the variance after each step sequentially, then we assign the desired assignment to the current subject with a probability greater than 0.5. We believe this procedure at least performs better than random assignment. We leave these as future research topics.
	
	Last but not least, it is possible that the theoretical analysis in Theorem~\ref{thm:adaptive} can be further improved. As mentioned in Section~\ref{thm:GOE}, the result in Theorem~\ref{thm:GOE} allows $\sigma\to 0$. If $p\to 0$ (sparse) in Theorem~\ref{thm:adaptive}, the current analysis does not show that the proposed design still achieves significant improvement. This could be a very interesting problem for further research. In general, it has proven difficult for researchers to obtain theoretical results on the designs of network data, due to the complexity of the problem and the lack of technical tools. In this paper, we introduce the technique of Lyapunov functions.  This technique could provide a feasible way of studying the properties of general designs in network data.
	
	\newpage
	
	\section{Appendix: Proofs}\label{sec:proof}
	
	\subsection{Proof of Theorem~\ref{thm:random}}
	
	In this proof, we briefly denote $T:=T_{\text{random}}$. We have
	\begin{align*}
	\|A (1-2T)\|^2 = \sum_{i=1}^n (A_{i*} (1-2T))^2
	\end{align*}
	and observe that the distributions of $A_{i*} (1-2T_{\text{random}})$ are identical for all $i\in[n]$. Without loss of generality, we consider $i=1$. By the definition of random design, $\ex[1-2T] = 0$. Furthermore, $A$ and $T$ are independent in the random design, so $\ex[A_{1*} (1-2T)] = 0$. Hence we have
	\begin{align*}
	\ex[(A_{1*} (1-2T))^2] = \text{var}[A_{1*} (1-2T)].
	\end{align*}
	We recall that $(T_{2m-1}, T_{2m}) = (0,1)$ or $(1,0)$ with equiprobability. By independence, we have
	\begin{align*}
	\text{var}[A_{1*} (1-2T)] =\sum_{m=1}^{n/2} \text{var}[A_{1,2m-1}(1-2T_{2m-1})+A_{1,2m}(1-2T_{2m})]
	\end{align*}
	For $m\in[[2,n/2]]$, the distributions of $A_{1,2m-1}(1-2T_{2m-1})+A_{1,2m}(1-2T_{2m})$ are identical. When $m=1$, we are in the special case that $A_{11}=1$. Hence, it suffices to consider the cases when $m=1$ and $m=2$. When $m=1$, we have
	\begin{align*}
	\text{var}[A_{11}(1-2T_{1})+A_{12}(1-2T_{2})] = \ex[(1-A_{12})^2] = 1-p.
	\end{align*}
	When $m=2$, we have
	\begin{align*}
	\text{var}[A_{13}(1-2T_{3})+A_{14}(1-2T_{4})] = \ex[(A_{13}-A_{14})^2] = 2p(1-p).
	\end{align*}
	Hence $\ex[(A_{1*} (1-2T))^2] = np(1-p)+(1-2p)(1-p)$, and $\ex[\|A (1-2T)\|^2] = n^2p(1-p)+n(1-2p)(1-p)$. Taking the limit, we have
	\begin{align*}
	\lim_{n\to\infty} \frac {\ex[\|A (1-2T_{\text{random}})\|^2]}{n^2}  = p(1-p)
	\end{align*}
	as desired.
	
	\subsection{Proof of Theorem~\ref{thm:adaptive}}
	
	\begin{proof}
		For $i\in [2m]$ and $j\in m$, let us define
		\begin{align*}
		Y_{ij} = A_{i,2j} - A_{i, 2j-1}
		\end{align*}
		If $i\ne 2j-1$ and $i\ne 2j$, i.e., neither $A_{i,2j}$ nor $A_{i,2j-1}$ is on the diagonal, we have
		\begin{align}\label{eq:distribution:of:Y:entries}
		Y_{ij} =
		\begin{cases}
		-1, &\text {with probability }p(1-p);\\
		0,  &\text {with probability }p^2+(1-p)^2;\\
		1, &\text {with probability }p(1-p)
		\end{cases}
		\end{align}
		We recall that $A^{(2m)}$ is the $2m\times 2m$ submatrix of $A$, $\tilde T_{2m} = {\mathbf 1}_{2m} - 2T_{1:2m}$, and define $Y_m = Y_{1:2m, m+1}\in\mathbb R^{2m}$. In this section, we use the notations
		\begin{align*}
		\tS_m := S_{2m} = A^{(2m)} \tilde T_{2m},\quad \tI_m := I_{2m} = \|A^{(2m)} \tilde T_{2m}\|.
		\end{align*}
		Now, as we define $I_{2m+1}=I_{2m}=\tI_m$ in~\eqref{eq:odd:I}, it suffices to show
		\begin{align}\label{eq:adaptive:design:limit:alt}
		\lim\sup_{n\to\infty} \frac{\ex[\tI_n^4]}{16n^4}\le p^2(1-p)^2 - \frac 1{8}(2b-1)(2-\sqrt{2}(2b-1))^{3/2} p^{5/2}(1-p)^{5/2},
		\end{align}
		which is equivalent to \eqref{eq:adaptive:design:limit}.
		As the entries of $Y_m$ follow the distribution of~\eqref{eq:distribution:of:Y:entries} independently, we have
		\begin{align*}
		\ex[\|Y_{m}\|^2] = 2m\var [Y_{1,m+1}] = 2m\ex [Y_{1,m+1}^2] = 4mp(1-p) .
		\end{align*}
		We have $\ex[\|Y_{m}\|^2] = 4mp(1-p)$. We also define
		\begin{align*}
		Z_{2m+1} = \sum_{i=1}^{2m} A_{2m+1, i} \tilde T_{i}\quad\text{and}\quad
		Z_{2m+2} = \sum_{i=1}^{2m} A_{2m+2, i} \tilde T_{i}.
		\end{align*}
		By definition, we have $Z_{2m+2}-Z_{2m+1} = \tilde T_{2m}^{\top}Y_m$. As $(T_{2i-1}, T_{2i}) = (1,0)$ or $(0,1)$, we can write
		\begin{align*}
		Z_{2m+1} = \sum_{j=1}^{m} (A_{2m+1, 2j} -  A_{2m+1, 2j-1}) \tilde T_{2i} = \sum_{j=1}^{m}Y_{2m+1, j}\tilde T_{2i}.
		\end{align*}
		By symmetry of $Y_{ij}$ $\tilde T_{2i}$, we have that $Z_{2m+1}$ shares the same distribution as $\sum_{j=1}^{m}Y_{2m+1, j}$. It is also clear that $Z_{2m+2}$ shares that same distribution.
		Hence we have $\ex[Z_{2m+1}] =\ex[Z_{2m+2}]= 0$ and
		\begin{align*}
		\ex[Z_{2m+1}^2] =\ex[Z_{2m+2}^2] = \ex\Big[\Big(\sum_{j=1}^{m}Y_{2m+1, j}\Big)^2\Big]
		= \sum_{j=1}^{m}\var[Y_{2m+1, j}] = m\ex[Y_{2m+1, 1}^2] = 2mp(1-p).
		\end{align*}
		In the $m+1$ step of our proposed procedure, we observe two new columns and new rows of the adjacency matrix, which will change the imbalance measurement. The square of the imbalance measurement in the $m+1$ step will be either
		\begin{align*}
		U_m = \|\tS_m+Y_m\|^2 + (Z_{2m+1}-1+A_{2m+1, 2m+2})^2 + (Z_{2m+2}+1-A_{2m+1, 2m+2})^2
		\end{align*}
		or
		\begin{align*}
		V_m = \|\tS_m-Y_m\|^2 + (Z_{2m+1}+1-A_{2m+1, 2m+2})^2 + (Z_{2m+2}-1+A_{2m+1, 2m+2})^2.
		\end{align*}
		Step 3 of the procedure indicates that our new design will pick the smaller one of the above two with probability $b>1/2$, and choose the larger one otherwise. By the symmetry of the distributions of $Y_m$, one can observe that this these two terms have the same expectation, and by direct calculation, we obtain
		\begin{align*}
		\ex[(Z_{2m+1}-1+A_{2m+1, 2m+2})^2] = \ex[(Z_{2m+2}+1-A_{2m+1, 2m+2})^2]=2mp(1-p)+1-p.
		\end{align*}
		\textbf{1. Upper bound of $\ex[\tI_n]$.} we We have the conditional expectation of $\tI_{m+1}^2$,
		\begin{align*}
		\ex[\tI_{m+1}^2|\tS_m]
		&\le \ex[ \|\tS_m+Y_m\|^2 + (Z_{2m+1}-1+A_{2m+1, 2m+2})^2 + (Z_{2m+2}+1-A_{2m+1, 2m+2})^2|\tS_m]\\
		& = \ex [ \|\tS_m\|^2 + 2\tS_m^{\top} Y_m + \|Y_m\|^2 + (Z_{2m+2}+1-A_{2m+1, 2m+2})^2  \\
		&\quad + (Z_{2m+2}+1-A_{2m+1, 2m+2})^2|\tS_m ]\\
		& = \|\tS_m\|^2 + 4mp(1-p)+2mp(1-p)+1-p+2mp(1-p)+1-p\\
		& = \|\tS_m\|^2 + 8mp(1-p) + 2(1-p).
		\end{align*}
		Therefore,
		$\ex[\tI_{m+1}^2 - \tI_m^2 | \tI_m^2 ] = 8mp(1-p) + 2(1-p)$. Hence
		\begin{align*}
		\ex[\tI_{m+1}^2 - \tI_m^2] = \ex[\ex[\tI_{m+1}^2-\tI_m^2|\tI_m^2] ] \le 8mp(1-p)+2(1-p).
		\end{align*}
		In the first stage, the imbalance measurement $\ex[\tI_1^2] = (1-A_{12})^2+(A_{21}-1)^2 = 2(1-p)$. Thus
		\begin{align*}
		\ex[\tI_n^2] \le \sum_{m=0}^{n-1} 8mp(1-p)+2(1-p) = 4n(n-1)p(1-p) + 2n(1-p)
		\le \Big( 2n\sqrt{p(1-p)} + \sqrt{\frac{1-p}{4p}} \Big)^2.
		\end{align*}
		By Jensen's inequality,
		\begin{align*}
		\ex[\tI_n]\le \sqrt{\ex[\tI_n^2]}\le 2n\sqrt{p(1-p)} + \sqrt{\frac{1-p}{4p}}.
		\end{align*}
		\textbf{2. Lower bound of $\ex[\tI_n^2]$.} With the upper bound of the first moment, we can derive the lower bound of the second moment.
		\begin{align*}
		\ex[\tI_{m+1}^2|\tS_m] &= \ex[B \min(U_m, V_m) + (1-B)\max(U_m, V_m)|\tS_m]\\
		& = \ex[U_m|\tS_m] - (2b-1) \ex[|U_m - V_m| |\tS_m]\\
		& = \|\tS_m\|^2 + 8mp(1-p) + 2(1-p) \\
		&\quad - (2b-1)\ex[| 2\tS_m^\top Y_m + 2(1-A_{2m+1, 2m+2})(Z_{2m+1}-Z_{2m+2}) | |\tS_m]
		\end{align*}
		Now we aim to find an upper bound of $\ex[| 2\tS_m^\top Y_m + 2(1-A_{2m+1, 2m+2})(Z_{2m+1}-Z_{2m+2}) | |\tS_m]$. By Jensen's inequality,
		\begin{align*}
		&(\ex[| 2\tS_m^\top Y_m + 2(1-A_{2m+1, 2m+2})(Z_{2m+1}-Z_{2m+2}) | |\tS_m])^2\\
		&\le 4\ex[ (\tS_m^\top Y_m + (1-A_{2m+1, 2m+2})(Z_{2m+1}-Z_{2m+2}))^2 |\tS_m ]\\
		& = 4\ex[( \tS_m^\top Y_m )^2|\tS_m] + 8\ex[(1-A_{2m+1, 2m+2})(Z_{2m+1}-Z_{2m+2})\tS_m^\top Y_m|\tS_m]\\
		& \quad +4 \ex[(1-A_{2m+1, 2m+2})(Z_{2m+1}-Z_{2m+2})^2 |\tS_m ]
		\end{align*}
		Now we will find the condition expectation of these three terms. As the entries of $Y_m$ are i.i.d. with distribution~\eqref{eq:distribution:of:Y:entries} we have
		\begin{align*}
		\ex[( \tS_m^\top Y_m )^2|\tS_m]  = \|\tS_m\|^2 \ex[Y_{21}^2] = 2\|\tS_m\|^2p(1-p).
		\end{align*}
		For the second term, by the definition of $Y_m, Z_{2m+1}$ and $Z_{2m+2}$, we have $Z_{2m+2}-Z_{2m+1}=Y_m^{\top}\tilde T_{2m}$. Hence
		\begin{align*}
		\ex[(1-A_{2m+1, 2m+2})(Z_{2m+1}-Z_{2m+2})\tS_m^\top Y_m|\tS_m]
		= (1-p)\ex[-Y_{m}^\top \tilde T_{2m} \tS_m^\top Y_m|\tS_m].
		\end{align*}
		As the distributions of the $(2i-1)$-th and $2i$-th rows are identical, we have $\pr(\tilde T_i=-1|\tS_m)=\pr(\tilde T_i=1|\tS_m)=0.5$. Hence for all $i$ and $m$, $T_i$ and $\tS_m$ are independent. $T_i$ and $\tS_m$ depend only on the submatrix $A^{(2m)}$, so they are independent of $Y_m$. Thus,
		\begin{align*}
		\ex[Y_{m}^\top \tilde T_{2m} \tS_m^\top Y_m|\tS_m]=0,
		\end{align*}
		which implies the second term vanishes. Using $Z_{2m+2}-Z_{2m+1}=\tilde T_{2m}^{\top} Y_m$ again, we have
		\begin{align*}
		\ex[(1-A_{2m+1, 2m+2})(Z_{2m+1}-Z_{2m+2})^2 |\tS_m ] &= (1-p)\ex[(\tilde T_{2m}^{\top} Y_m)^2] \\
		&= (1-p)\|\tilde T_{2m}\|^2 \ex[Y_{21}^2]\\
		&=4mp(1-p)^2.
		\end{align*}
		Therefore, using the fact that $\sqrt{x+y}\le \sqrt x+ \sqrt y$ for $x,y\ge 0$, we have
		\begin{align*}
		\ex[| 2\tS_m^\top Y_m + 2(1-A_{2m+1, 2m+2})(Z_{2m+1}-Z_{2m+2}) | |\tS_m]
		&\le 2\sqrt{2\|\tS_m\|^2p(1-p)+2mp(1-p)^2}\\
		&\le 2\|\tS_m\|\sqrt{2p(1-p)} + 4\sqrt{mp}(1-p).
		\end{align*}
		As $\|\tS_m\|=\tI_m$, we have
		\begin{align*}
		\mathbb E[\tI_{m+1}^2-\tI_m^2]
		&= \ex[\mathbb E[\tI_{m+1}^2 - \tI_m^2 | \tS_m]] \\
		&\ge 8mp(1-p) +2(1-p) - 2(2b-1)(\ex[\tI_m]\sqrt{2p(1-p)} + 2\sqrt{2mp}(1-p))\\
		&= 8mp(1-p) +2(1-p) -4\sqrt{2} (2b-1)mp(1-p)-\sqrt2(2b-1)p(1-p) \\
		&\quad - 4\sqrt{mp}(2b-1)(1-p)\\
		& = (8-4\sqrt 2(2b-1))mp(1-p) - (4\sqrt{mp}(2b-1) + 2 - \sqrt 2(2b-1)p)(1-p).
		\end{align*}
		Recalling that $\ex[\tI_1^2]=2(1-p)$, and using $\sum_{m=1}^{n-1}\sqrt m\ge \frac 23(n-1)^{3/2}$, we have
		\begin{align*}
		\mathbb E[\tI_{n}^2]
		&=2(1-p)+\sum_{m=1}^{n-1}   (8-4\sqrt 2(2b-1))mp(1-p) - (4\sqrt{mp}(2b-1) + 2 - \sqrt 2(2b-1)p)(1-p)\\
		&\ge (4-2\sqrt 2(2b-1))n(n-1)p(1-p)- \frac{8}{3}(2b-1)\sqrt p(1-p)(n-1)^{3/2} \\
		&\quad- (2-\sqrt 2 (2b-1)p)(n-1)(1-p).
		\end{align*}
		\textbf{3. Lower bound of $\ex[\tI_n^3]$.} Combining with Jensen's inequality, we have
		\begin{align*}
		\ex[\tI_n^3]\ge
		\ex[\tI_n^2]^{3/2}
		\ge (4-2\sqrt 2(2b-1))n^2p(1-p) + O(n^{3/2})\sqrt{p}(1-p))^{3/2}.
		\end{align*}
		Since $(x+y)^{3/2}\ge x^{3/2}+y^{3/2}$ for $x,y\ge 0$, we have
		\begin{align*}
		\ex[\tI_n^3]\ge(4-2\sqrt 2(2b-1))^{3/2}n^3p^{3/2}(1-p)^{3/2} + O(n^{9/4}p^{3/4}(1-p)^{3/2}).
		\end{align*}
		\textbf{4. Upper bound of $\ex[\tI_n^4]$.} Now we are ready to establish the upper bound of the fourth moment of $\tI_n$. We have
		\begin{align}\label{eq:fourth:moment}
		\begin{split}
		\mathbb E[\tI_{m+1}^4 ]
		&= \ex[ B\max(U_m^2, V_m^2) + (1-B) \min(U_m^2, V_m^2)  ]\\
		&= \ex[U_m^2] - (2b-1)\ex[|U_m^2 - V_m^2| ].
		\end{split}
		\end{align}
		where the first term
		\begin{align*}
		\ex[U_m^2] =\ex[(\|\tS_m+Y_m\|^2 + (Z_{2m+1}-1+A_{2m+1, 2m+2})^2 + (Z_{2m+2}+1-A_{2m+1, 2m+2})^2)^2].
		\end{align*}
		As $\ex[\tS_m]=0$ and $\tS_m$ is independent of $Y_m$, $Z_{2m+1}$ and $Z_{2m+2}$, all of the cross terms containing $S_{m}^\top Y_m$ have expectation 0. Using $\ex[\tS_m^\top Y_m]=0$, it is easy to check
		\begin{align*}
		\ex[(\tS_m^\top Y_m)^2] = \var(\tS_m^\top Y_m) = 2p(1-p)\ex[\|\tS_m\|^2].
		\end{align*}
		By independence again, we have
		\begin{align*}
		\ex[2\|\tS_m\|^2\|Y_m\|^2] = 2\ex[\|\tS_m\|^2] \ex[\|Y_m\|^2] = 8mp(1-p)\ex[\|\tS_m\|^2],
		\end{align*}
		and recalling that $\ex[(Z_{2m+1}-1+A_{2m+1, 2m+2})^2] = \ex[(Z_{2m+2}+1-A_{2m+1, 2m+2})^2]=2mp(1-p)+1-p$, we have
		\begin{align*}
		\ex[2\|\tS_m\|^2(Z_{2m+1}-1+A_{2m+1, 2m+2})^2] &= \ex[\|\tS_m\|^2(Z_{2m+2}+A_{2m+2, 2m+1}-1)^2] \\
		&=2 (2mp+1)(1-p)\ex[\|\tS_m\|^2].
		\end{align*}
		The other terms do not contain $\tS_m$. We first compute the fourth moments:
		\begin{align*}
		\ex[\|Y_m\|^4] &=\ex\Big[\Big(\sum_{i=1}^{2m}Y_{i,m+1}\Big)^4\Big]
		= \sum_{i=1}^{2m} \ex [Y_{i,m+1}^4] + \sum_{1\le i<j\le m} \ex [Y_{i,m+1}^2Y_{j,m+1}^2] \\
		&= 4mp(1-p) + 2m(2m-1)p^2(1-p)^2,
		\end{align*}
		and we recall that $Z_{2m+1}$ has the same distribution as $\sum_{i=1}^m Y_{2m+1, i}$
		\begin{align*}
		\ex[Z_{2m+1}^4]&=\sum_{i=1}^m \ex[Y_{2m+1,i}^4] + {4\choose 2}\sum_{1\le i<j\le m} \ex[Y_{2m+1,i}^2Y_{2m+1,j}^2]\\
		&=2mp(1-p)+12m(m-1)p^2(1-p)^2.
		\end{align*}
		We have $\ex[Z_{2m+2}^4] = 2mp(1-p)+12m(m-1)p^2(1-p)^2$ in the same manner.
		By symmetry of $Z_{2m+1}$, we have $\ex[Z_{2m+1}] = \ex[Z_{2m+1}^3] = 0$, so
		\begin{align*}
		&\ex[(Z_{2m+1}-1+A_{2m+1, 2m+2})^4]\\
		& = \ex[Z_{2m+1}^4 + 6Z_{2m+1}^2(A_{2m+1, 2m+2}-1)^2+(A_{2m+1, 2m+2}-1)^4]\\
		& =  2mp(1-p)+12m(m-1)p^2(1-p)^2 + 6(2mp(1-p))(1-p) + 1-p\\
		& \le (14mp+1)(1-p) + 12m^2p^2(1-p)^2.
		\end{align*}
		Applying these bounds, we have
		\begin{align*}
		&\ex[U_m^2] =\ex[(\|\tS_m+Y_m\|^2 + (Z_{2m+1}-1+A_{2m+1, 2m+2})^2 + (Z_{2m+2}+1-A_{2m+1, 2m+2})^2)^2]\\
		& = \ex[\|\tS_m\|^4 + (2\tS_m^\top Y_m)^2 + \|Y_m\|^4 + (Z_{2m+1}-1+A_{2m+1, 2m+2})^4+ (Z_{2m+2}+1-A_{2m+1, 2m+2})^4\\
		& \quad + 2\|\tS_m\|^2\|Y_m\|^2 + 2\|\tS_m\|^2(Z_{2m+1}-1+A_{2m+1, 2m+2})^2 + 2\|\tS_m\|^2(Z_{2m+2}+1-A_{2m+1, 2m+2})^2\\
		& \quad + 2\|Y_m\|^2 (Z_{2m+1}-1+A_{2m+1, 2m+2})^2 + 2\|Y_m\|^2(Z_{2m+2}+1-A_{2m+1, 2m+2})^2\\
		&\quad + 2(Z_{2m+1}-1+A_{2m+1, 2m+2})^2(Z_{2m+2}+1-A_{2m+1, 2m+2})^2]\\
		& = \ex[\|\tS_m\|^4] + 4p(1-p)\ex[\|\tS_m\|^2] +4mp(1-p) + 2m(2m-1)p^2(1-p)^2 \\
		&\quad + 2(2mp(1-p)+12m(m-1)p^2(1-p)^2) + 2\ex[\|\tS_m\|^2](4mp(1-p)) \\
		&\quad + 4\ex[\|\tS_m\|^2](2mp(1-p)+1-p) + 4(4mp(1-p))(2mp(1-p)+1-p)\\
		&\quad +4(2mp(1-p)+1-p)^2\\
		& = \ex[\|\tS_m\|^4] + 16mp(1-p)\ex[\|\tS_m\|^2] + O((m^2p(1-p)+\ex[\|\tS_m\|^2])\\
		& = \ex[\|\tS_m\|^4] + 16mp(1-p)\ex[\|\tS_m\|^2] + O(m^2p(1-p)).
		\end{align*}
		We denote as ``high-order terms" those whose expected values have an order of at most $O((m^2p(1-p)+\ex[\|\tS_m\|^2])p(1-p))$. Now let us consider the terms in $U_m^2 - V_m^2$ with absolute value. The only term that does not belong to the high-order terms is $4(2b-1)\ex[\|\tS_m\|^2|\tS_m^\top Y_m|]$. In other words,
		\begin{align*}
		|\ex[|U_m^2 - V_m^2|] - 4(2b-1)\ex[\|\tS_m\|^2|\tS_m^\top Y_m|]| = O(m^2p(1-p)).
		\end{align*}
		By Lemma~\ref{lem:kahane}, we have
		\begin{align*}
		\ex[\|\tS_m\|^2|\tS_m^\top Y_m|]
		= \ex[\ex[\|\tS_m\|^2|\tS_m^\top Y_m|] |\tS_m]
		\ge \ex[2p(1-p)\|\tS_m\|^3]
		= 2p(1-p) \ex[\tI_m^3].
		\end{align*}
		Applying this inequality to~\eqref{eq:fourth:moment}, we have
		\begin{align*}
		&\mathbb E[\tI_{m+1}^4]
		\le \ex[\tI_m^4] + 16mp(1-p)\ex[\tI_m^2] - 8(2b-1)p(1-p)\ex[\tI_m^3] + O((m^2p(1-p)+\ex[\tI_m^2])p(1-p))\\
		&\le \ex[\tI_m^4] + 16mp(1-p)(4m^2p(1-p)+2m(1-p)) +O(m^2p(1-p))\\
		&\quad - 8(2b-1)p(1-p)[(4-2\sqrt 2(2b-1))^{3/2}m^3p^{3/2}(1-p)^{3/2} + O(m^{9/4}p^{3/4}(1-p)^{3/2})]\\
		&=\ex[\tI_m^4] + 64m^3p^2(1-p)^2 - 8(2b-1)(4-2\sqrt{2}(2b-1))^{3/2} m^3 p^{5/2}(1-p)^{5/2} + O(m^{9/4}p(1-p)).
		\end{align*}
		In the first stage, the imbalance measurement $\ex[\tI_1^4] = (1-A_{12})^4+(A_{21}-1)^4 = 2(1-p)$. Thus
		\begin{align*}
		\ex[\tI_n^4]
		& = 2(1-p) + \sum_{m=1}^{n-1} \ex[\tI_{m+1}^4 - \tI_m^4]\\
		& = 2(1-p) + \sum_{m=1}^{n-1}  64m^3p^2(1-p)^2 - 8(2b-1)(4-2\sqrt{2}(2b-1))^{3/2} m^3p^{5/2}(1-p)^{5/2}\\
		& \quad  + O(m^{9/4}p(1-p))\\
		& = (16 p^2(1-p)^2 - 2(2b-1)(2-\sqrt{2}(2b-1))^{3/2} p^{5/2}(1-p)^{5/2} )n^4 + O(n^{13/4}p(1-p)).
		\end{align*}
		Therefore,
		\begin{align*}
		\lim\sup_{n\to\infty} \frac{\ex[\tI_n^4]}{n^4}\le 16 p^2(1-p)^2 - 2(2b-1)(2-\sqrt{2}(2b-1))^{3/2} p^{5/2}(1-p)^{5/2} ,
		\end{align*}
		which proves \eqref{eq:adaptive:design:limit:alt}.
	\end{proof}
	
	\subsection{Proof of Theorem~\ref{thm:GOE}}
	The outline of the proof is very similar to that of Theorem~\ref{thm:adaptive}. We adopt the definitions of $\tI_m$, $\tS_m$, $Y_{ij}$, $Y_m$, $Z_{2m+1}$ and $Z_{2m+2}$, $U_m$ and $V_m$. As we replace the Erd\H os-R\' enyi random graph model by the GOE, we have $A_{ij}\sim \mathcal (0, \sigma^2)$ for $1\le i\le j\le n$ and $A_{ij} = A_{ji}$ if $1\le i<j\le n$. Then $Y_{ij}\sim \mathcal N(0, 2\sigma^2)$, $Z_{2m+1}, Z_{2m+2}\sim \mathcal N(0, 2m\sigma^2)$. Hence
	\begin{align*}
	\|Y_{ij}\|^2 \sim 2\sigma^2 \chi^2_{2m}
	\quad\text{and} \quad
	Z_{2m+1}^2, Z_{2m+2}^2\sim 2m\sigma^2\chi_1^2.
	\end{align*}
	We need to replace the moments of these variables in the proof of Theorem~\ref{thm:adaptive}. \\
	
	\textbf{1. Upper bound of $\ex[\tI_n]$. } The conditional expectation of $\tI_{m+1}^2$ is bounded by
	\begin{align*}
	\ex[\tI_{m+1}^2|\tS_m]
	&\le \ex[ \|\tS_m+Y_m\|^2 + (Z_{2m+1}-1+A_{2m+1, 2m+2})^2 + (Z_{2m+2}+1-A_{2m+1, 2m+2})^2|\tS_m]\\
	& = \ex [ \|\tS_m\|^2 + 2\tS_m^{\top} Y_m + \|Y_m\|^2 + (Z_{2m+2}+1-A_{2m+1, 2m+2})^2  \\
	&\quad + (Z_{2m+2}+1-A_{2m+1, 2m+2})^2|\tS_m ]\\
	& = \|\tS_m\|^2 + 4m\sigma^2 + (2m+1)\sigma^2+(2m+1)\sigma^2 + 2\\
	& = \|\tS_m\|^2 + (8m+2)\sigma^2 + 2.
	\end{align*}
	Therefore,
	$\ex[\tI_{m+1}^2 - \tI_m^2 | \tI_m^2 ] \le (8m+2)\sigma^2 + 2$. Hence
	\begin{align*}
	\ex[\tI_{m+1}^2 - \tI_m^2] = \ex[\ex[\tI_{m+1}^2-\tI_m^2|\tI_m^2] ] \le (8m+2)\sigma^2 + 2.
	\end{align*}
	In the first stage, $\ex[\tI_1^2]\le 2+2\sigma^2$. Thus,
	\begin{align*}
	\ex[\tI_n^2]\le 2+2\sigma^2+\sum_{m=1}^{n-1} (8m+2)\sigma^2 + 2\le 4n^2\sigma^2+2n\le \Big( 2n\sigma+\frac{1}{2\sigma} \Big)^2.
	\end{align*}
	By Jensen's inequality,
	\begin{align*}
	\ex[\tI_n]\le \sqrt{\ex[\tI_n^2]}\le 2n\sigma+\frac 1{2\sigma}.
	\end{align*}
	\textbf{Lower bound of $\ex[\tI_n^2]$.} As we did in the proof of Theorem~\ref{thm:adaptive}, to find the lower bound of $\ex[\tI_{m+1}^2-\tI_m^2]$, we need to find the upper bound for
	\begin{align*}
	&(\ex[| 2\tS_m^\top Y_m + 2(1-A_{2m+1, 2m+2})(Z_{2m+1}-Z_{2m+2}) | |\tS_m])^2\\
	&
	\le\ex[4( \tS_m^\top Y_m )^2+8(1-A_{2m+1, 2m+2})(Z_{2m+1}-Z_{2m+2})\tS_m^\top Y_m\\
	&\quad +4(1-A_{2m+1, 2m+2})(Z_{2m+1}-Z_{2m+2})^2 |\tS_m ].
	\end{align*}
	The second term has expectation 0 for the same reason as in the proof of Theorem~\ref{thm:adaptive}. Since $Y_m\sim \mathcal{N}(0, 2\sigma^2\mathbf \tI_m)$, we have $\ex[(\tS_m^\top Y_m )^2|\tS_m]=2\sigma^2\|\tS_m\|^2$.  We recall that $Z_{2m+1}, Z_{2m+2}\sim \mathcal N(0, 2m\sigma^2)$, so
	\begin{align*}
	\ex[(1-A_{2m+1, 2m+2})(Z_{2m+1}-Z_{2m+2})^2 |\tS_m ] = 4m\sigma^2.
	\end{align*}
	Applying $Z_{2m+1}-Z_{2m+2} = \ones_{2m}^\top Y_m$, we have
	\begin{align*}
	&\ex[| 2\tS_m^\top Y_m + 2(1-A_{2m+1, 2m+2})(Z_{2m+1}-Z_{2m+2}) | |\tS_m]\\
	&= 2\ex[(\tS_m+(1-A_{2m+1, 2m+2})\ones_{2m})^\top Y_m|\tS_m]\\
	& = 2\sqrt{2/ \pi}\ex[\|\tS_m+(1-A_{2m+1, 2m+2})\ones_{2m}\||\tS_m]\\
	& \ge 2\sqrt{2/ \pi}(\|\tS_m\|+\sqrt{2m(1+\sigma^2)})\sigma.
	\end{align*}
	where in the last equality, we use if $X\sim\mathcal N(0,\sigma^2)$, then $\ex[|X|]=\sigma\sqrt{2/\pi}$.
	Using the same definition of $U_m$ and $V_m$ from the proof of Theorem~\ref{thm:adaptive}, we have
	\begin{align*}
	\ex[\tI_{m+1}^2|\tS_m] &= \ex[B \min(U_m, V_m) + (1-B)\max(U_m, V_m)|\tS_m]\\
	& = \ex[U_m|\tS_m] - (2b-1) \ex[|U_m - V_m| |\tS_m]\\
	& = \|\tS_m\|^2 + 8m\sigma^2+2 \\
	&\quad - (2b-1)\ex[| 2\tS_m^\top Y_m + 2(1-A_{2m+1, 2m+2})(Z_{2m+1}-Z_{2m+2}) | |\tS_m]\\
	&\ge \|\tS_m\|^2 + 8m\sigma^2+2-2(2b-1)\sqrt{2/ \pi}(\|\tS_m\|+\sqrt{2m(1+\sigma^2)})\sigma\\
	&\ge \|\tS_m\|^2 + 8m\sigma^2+2-2(2b-1)\sqrt{2/ \pi}(2m\sigma^2 + \sqrt{2m(1+\sigma^2)}\sigma+0.5)\\
	& = \|\tS_m\|^2+(8-4(2b-1)\sqrt{2/\pi})m\sigma^2 + O(\sqrt m\sigma).
	\end{align*}
	Using $\sum_{m=1}^{n-1}\sqrt m\ge \frac 23(n-1)^{3/2}$, we have
	\begin{align*}
	\mathbb E[\tI_{n}^2]
	& = \ex[\tI_1^2] + \sum_{i=1}^{n-1} \ex[\ex[\tI_{m+1}^2-\tI_m^2|\tS_m]]\\
	& \ge 2+2\sigma^2 +\sum_{m=1}^n (8-4(2b-1)\sqrt {2/\pi})m\sigma^2 + O(\sqrt m\sigma)\\
	&\ge (4-2\sqrt {2/\pi}(2b-1))n^2\sigma^2 + O(n^{3/2}\sigma).
	\end{align*}
	\textbf{Lower bound of $\ex[\tI_n^3]$.} By Jensen's inequality and the fact that $(x+y)^{3/2}\ge x^{3/2}+y^{3/2}$, we have
	\begin{align*}
	\ex[\tI_n^3]\ge \ex[\tI_n^2]^{3/2}\ge (4-\sqrt {2/\pi}(2b-1))^{3/2}n^3\sigma^3 + O(n^{9/4}\sigma^{3/2}).
	\end{align*}
	\textbf{Upper bound of $\ex[\tI_n^4]$.} We have
	\begin{align*}
	\mathbb E[\tI_{m+1}^4 ]
	&= \ex[ B\max(U_m^2, V_m^2) + (1-B) \min(U_m^2, V_m^2)  ]\\
	&= \ex[U_m^2] - (2b-1)\ex[|U_m^2 - V_m^2| ].
	\end{align*}
	from \eqref{eq:fourth:moment}. Using the same arguments in the proof of Theorem~\ref{thm:adaptive}, we have
	\begin{align*}
	\ex[U_m^2] = \ex[\|\tS_m\|^4]+16m\sigma^2\ex[\|\tS_m\|^2] + O(m^2\sigma^2),
	\end{align*}
	and
	\begin{align*}
	\ex[|U_m^2 - V_m^2| ] = 4(2b-1)\ex[\|\tS_m\|^2|\tS_m^\top Y_m|]+O(m^2\sigma^2).
	\end{align*}
	Using the expectation of a folded normal random variable again, we have
	\begin{align*}
	\ex[\|\tS_m\|^2|\tS_m^\top Y_m|] &= \ex[\|\tS_m\|^2|\tS_m^\top Y_m|]\\
	& = \ex[\ex[\|\tS_m\|^2|\tS_m^\top Y_m|]|\tS_m]\\
	& = \sqrt{2/\pi}\sigma \ex[\tI_m^3].
	\end{align*}
	We apply the upper bound of $\ex[\tI_m^2]$ and lower bound of $\ex[\tI_m^3]$, and have
	\begin{align*}
	\mathbb E[\tI_{m+1}^4 ]
	&\le \ex[\tI_m^4]+ 64m^3\sigma^4 - 4(2b-1)\sqrt{2/\pi}(4-\sqrt {2/\pi}(2b-1))^{3/2}m^3\sigma^4 \\
	&\quad +O(m^{9/4}\sigma^{5/2}+m^2\sigma^2).
	\end{align*}
	Therefore, we have
	\begin{align*}
	\ex[\tI_n^4]
	& = \ex[\tI_1^4] + \sum_{m=1}^{n-1} \mathbb E[\tI_{m+1}^4 - \tI_m^4] \\
	& \le (16- 4(2b-1)\sqrt{2/\pi}(4-\sqrt {2/\pi}(2b-1))^{3/2})m^4\sigma^4 + O(m^{13/4}\sigma^{5/2}+m^3\sigma^2)
	\end{align*}
	Assuming $n\sigma^2\to \infty$, we have
	\begin{align*}
	O\Big( \frac{ m^{13/4}\sigma^{5/2}+m^3\sigma^2}{m^4\sigma^4}\Big) = O((n\sigma^2)^{-3/4}+(n\sigma^2)^{-1})\to 0.
	\end{align*}
	Hence
	\begin{align*}
	\lim\sup_{n\to\infty} \frac{\ex[\tI_n^4]}{m^4\sigma^4 }
	\le (16- 4(2b-1)\sqrt{2/\pi}(4-\sqrt {2/\pi}(2b-1))^{3/2}),
	\end{align*}
	as desired.
	
	\subsection{Auxiliary Lemmas}
	\begin{lemma}[Khinchin-Kahane inequality]\label{lem:kahane}
		For $i\in[n]$, let $Y_i= -1, 0, 1$ with probability $p_i, 1-2p_i, p_i$ identically and independently distributed for $p_i\in(0, 1/2)$. Then we have
		\begin{align*}
		\min_{i\in[n]}2p_i \le \inf_{x\in S^{n-1}} \ex[|x^\top Y|]
		\le\sup_{x\in S^{n-1}} \ex[|x^\top Y|]\le \max_{i\in[n]}\sqrt{2p_i},
		\end{align*}
		where $S^{n-1} = \{x\in \mathbb R^n: \|x\|=1\}$.
	\end{lemma}
	
	\begin{proof}
		{\bf Upper bound.} For $\|x\|=1$, we have
		\begin{align*}
		\ex[|x^\top Y|]^2\le \ex[|x^\top Y|^2] = \sum_{i=1}^n x_i^2 \ex[|Y_i|^2] = \|x\|^2\max_{i\in[n]}\ex[Y_i^2]=\max_{i\in[n]}\ex[Y_i^2].
		\end{align*}
		Hence $\ex[|x^\top Y|]\le \max_{i\in[n]}\sqrt {\ex[Y_i^2]}=\max_{i\in[n]}\sqrt{2p_i}$. \\
		
		{\bf Lower Bound.}
		Let us define $S^{n-1}_+ = \{x\in S^{n-1}: \forall i\in[n], x_i\ge 0 \}$.
		By symmetry of $Y_i$, if suffices to consider the infimum for $x$ over $S^{n-1}_+$ to avoid loss of generality. We claim that if $j=\arg\min_{i\in[n]}2p_i$, then the minimum is achieved at $x = e_j$. We will prove this by induction. The claim is clearly correct when $n=1$. Now let us consider the case $n+1$, given the statement is true for $n$. In other words, it suffices to show that
		\begin{align*}
		\inf_{x\in S^{n}} \ex[|x^\top Y|]=
		\inf_{x\in S^{n-1}}\inf_{\theta\in[0,\pi/2]}\ex[|x^\top Y\cos\theta+ Y_{n+1}\sin\theta|]\ge \min_{i\in[n+1]}2p_i.
		\end{align*}
		given
		\begin{align*}
		\inf_{x\in S^{n-1}} \ex[|x^\top Y|]\ge \min_{i\in[n]}2p_i.
		\end{align*}
		We note that for $x\in S^{n-1}_+$, $\|(x\cos\theta, \sin\theta)\|^2 = \|x\|^2\cos^2\theta +\sin^2\theta = 1$. We have
		\begin{align*}
		\ex[|x^\top Y\cos\theta  + Y_{n+1}\sin\theta|]
		=& \sum_{y=-1}^1 \ex[|x^\top Y\cos\theta  + Y_{n+1}\sin\theta| | Y_{n+1}=y] \pr(Y_{n+1} = y)\\
		=& (1-2p_{n+1} )\cos\theta\ex[|x^\top Y|] + p_{n+1} \ex[|x^\top Y\cos\theta+\sin\theta|]\\
		&+ p_{n+1}\ex[|x^\top Y\cos\theta-\sin\theta|].
		\end{align*}
		By symmetry of $x^\top Y$, we have
		\begin{align*}
		\ex[|x^\top Y\cos\theta-\sin\theta|] = \ex[|-x^\top Y\cos\theta-\sin\theta|]=\ex[|x^\top Y\cos\theta+\sin\theta|].
		\end{align*}
		By Lemma~\ref{lem:exp:abs}, we have
		\begin{align*}
		\ex[|x^\top Y\cos\theta+\sin\theta|] = \ex[\max\{|x^\top Y|\cos\theta, \sin\theta\}]
		\ge \max\{ \ex[|x^\top Y|]\cos\theta, \sin\theta \}.
		\end{align*}
		Therefore,
		\begin{align*}
		\ex[|x^\top Y\cos\theta  + Y_{n+1}\sin\theta|]&=
		(1-2p_{n+1} )\ex[|x^\top Y|]\cos\theta + 2p_{n+1} \ex[|x^\top Y\cos\theta+\sin\theta|]\\
		&\ge (1-2p_{n+1} )\ex[|x^\top Y|]\cos\theta + 2p_{n+1} \max\{ \ex[|x^\top Y]\cos\theta, \sin\theta \}\\
		& = \max\{ \ex[|x^\top Y|]\cos\theta, (1-2p_{n+1})\ex[|x^\top Y|]\cos\theta+2p_{n+1}\sin\theta \}.
		\end{align*}
		The last line is a concave function corresponding to the variable $\theta\in[0,\pi/2]$. For every $x\in S^{n-1}_+$, it achieves the minimum when either $\theta=0$ or $\theta=1$. Thus for every $x\in S^n$ and $\theta\in[0,\pi/2]$,
		\begin{align*}
		\ex[|x^\top Y\cos\theta  + Y_{n+1}\sin\theta|]&\ge
		\min\{ \max\{ \ex[|x^\top Y|], (1-2p_{n+1} )\ex[|x^\top Y|] \}, \max\{ 0, 2p_{n+1}  \} \} \\
		&= \min\{ \ex[|x^\top Y|], 2p_{n+1}  \}.
		\end{align*}
		By the inductive assumption, $\ex[|x^\top Y|]\ge \min_{i\in[n]} 2p_i$, so
		\begin{align*}
		\inf_{x\in S^{n}} \ex[|x^\top Y|]=
		\inf_{x\in S^{n-1}}\inf_{\theta\in[0,\pi/2]}\ex[|x^\top Y\cos\theta+ Y_{n+1}\sin\theta|]\ge \min_{i\in[n+1]}2p_i,
		\end{align*}
		which finishes the proof.
	\end{proof}

	\begin{lemma}\label{lem:exp:abs}
		Suppose $Y$ is a symmetric random variable and $x\ge 0$ is fixed, then
		\begin{align*}
		\ex[|Y+x|]\ge \max\{ \ex[|Y|], x \}.
		\end{align*}
	\end{lemma}
	
	\begin{proof}
		We assume $Y$ is discrete. Other cases simply follow from the arguments below.
		\begin{align*}
		\ex[|Y+x|] &=\sum_y |y+x| \pr(Y=y) = x\pr(Y=0) + \sum_{y>0} (y+x+|-y+x|)\pr (Y=y)\\
		&\ge x\pr(Y=0) + \sum_{y>0} (y+x+y-x)\pr (Y=y)\ge \sum_{y>0} 2y\pr(Y=y)=\ex[|Y|].
		\end{align*}
		Additionally,
		\begin{align*}
		\ex[|Y+x|] &=x\pr(Y=0) + \sum_{y>0} (y+x+|-y+x|)\pr (Y=y)\\
		&\ge x\pr(Y=0) + \sum_{y>0} (y+x-x+y)\pr (Y=y)=x.
		\end{align*}
		The proof is complete.
	\end{proof}

	\bibliographystyle{imsart-number}
	\bibliography{refs}
	
\end{document}